\documentclass[review]{elsarticle}
\usepackage[utf8]{inputenc}
\usepackage{amsthm}
\usepackage{amsmath}
\usepackage{amsfonts}
\usepackage{amssymb}
\usepackage{graphicx}
\usepackage{color}
\usepackage[normalem]{ulem}
\usepackage{multirow}
\usepackage{mathtools}
\usepackage{xcolor}
\usepackage{a4wide}
\usepackage{lineno,hyperref}
\newtheorem{theorem}{Theorem}
\newtheorem{lemma}[theorem]{Lemma}

\newdefinition{remark}{Remark}
\newdefinition{definition}{Definition} \newdefinition{example}{Example}
\usepackage{wrapfig}
\usepackage{algorithm,algpseudocode}
\usepackage{mathtools}

\usepackage{graphicx}
\usepackage{tabularx}
\usepackage{textcomp}
\usepackage{xcolor}
\usepackage[labelformat=simple]{subcaption}

\usepackage{booktabs}

\usepackage{pgf,tikz,pgfplots}
\pgfplotsset{compat=1.15}
\usepackage{mathrsfs}
\usepackage{hyperref}
\usepackage{multicol}

\usepackage[frak=esstix]{mathalpha}

\usepackage[linesnumbered,ruled,vlined,algo2e]{algorithm2e}
\usepackage{algcompatible}
\newcommand{\tarc}{\mbox{\large$\frown$}}
\newcommand{\arc}[1]{\stackrel{\tarc}{#1}}

\DeclareFontFamily{OMX}{yhex}{}
\DeclareFontShape{OMX}{yhex}{m}{n}{<->yhcmex10}{}
\DeclareSymbolFont{yhlargesymbols}{OMX}{yhex}{m}{n}
\DeclareMathAccent{\wideparen}{\mathord}{yhlargesymbols}{"F3}

\newcounter{manualcase}
\newcounter{manualsubcase}[manualcase]

\usepackage{xcolor}
\usepackage[linesnumbered,ruled,vlined,algo2e]{algorithm2e}

\SetCommentSty{mycommfont}

\usepackage{amssymb}

\usepackage{a4wide}
\begin{document}

\begin{frontmatter}



\title{Uniform Partitioning of a Bounded Region using Opaque ASYNC Luminous Mobile Robots}


\author{Subhajit Pramanick}
\ead{subhajit.pramanick@iitg.ac.in}
\author{Saswata Jana}
\ead{saswatajana@iitg.ac.in}
\author{Adri Bhattacharya}
\ead{a.bhattacharya@iitg.ac.in}
\author{Partha Sarathi Mandal}
\ead{psm@iitg.ac.in}
\affiliation{organization={Department of Mathematics},
            addressline={Indian Institute of Technology Guwahati}, 
            city={Guwahati},
            postcode={781039}, 
            state={Assam},
            country={India}}

\begin{abstract}
We are given $N$ autonomous mobile robots inside a bounded region. The robots are opaque which means that three collinear robots are unable to see each other as one of the robots acts as an obstruction for the other two. They operate in classical \emph{Look-Compute-Move} (LCM) activation cycles. Moreover, the robots are oblivious except for a persistent light (which is why they are called \emph{Luminous robots}) that can determine a color from a fixed color set. Obliviousness does not allow the robots to remember any information from past activation cycles. The Uniform Partitioning problem requires the robots to partition the whole region into sub-regions of equal area, each of which contains exactly one robot. Due to application-oriented motivation, we, in this paper consider the region to be well-known geometric shapes such as rectangle, square and circle. We investigate the problem in \emph{asynchronous} setting where there is no notion of common time and any robot gets activated at any time with a fair assumption that every robot needs to get activated infinitely often. To the best of our knowledge, this is the first attempt to study the Uniform Partitioning problem using oblivious opaque robots working under asynchronous settings. We propose three algorithms considering three different regions: rectangle, square and circle. 
The algorithms proposed for rectangular and square regions run in $O(N)$ epochs whereas the algorithm for circular regions runs in $O(N^2)$ epochs, where an epoch is the smallest unit of time in which all robots are activated at least once and execute their LCM cycles. The algorithms for the rectangular, square and circular regions require $2$ (which is optimal), $5$ and $8$ colors, respectively.
\end{abstract}


\end{frontmatter}



\section{Introduction}
\subsection{Motivation}
Distributed algorithms for many real-world problems grabbed a lot of attention from researchers for many years now. A swarm of mobile robots is a very important tool in designing such distributed algorithms. The collaborative actions of these robots achieve the end goal of the problems. One such problem is \emph{Uniform Partitioning} of a bounded region, which is a very common problem in our daily life. For example, if a number of people are asked to paint a wall, the natural strategy is to divide the region into equal parts and assign each person to a distinct part of the wall for painting. Keeping the same motivation in mind, autonomous mobile robots can be more useful in much more critical situations like cleaning spillage of liquid radioactive waste in a laboratory. In such hazardous situations, robots are the safest options for us. Uniform partitioning is equally applicable to another scenario where a city needs to be well-covered with networks. A group of autonomous drones, each of which is equipped with necessary instruments, can be deployed to serve the purpose, where drones should position themselves in such a way that each of them covers a part of the whole city of the same area as others.

Classically, these robots are autonomous (no external control), homogeneous (execute the same algorithm), anonymous (having no unique identifier) and disoriented (do agree on any global coordinate system or orientation). The robots are modelled as points on the plane. They are equipped with vision which enables them to gather information from the surroundings. These robots operate in \emph{Look-Compute-Move} cycles, which we define later in the paper. Two robots might not be able to see each other due to the presence of other robots between them (it is called \emph{obstructed visibility model}). These robots are called \emph{opaque} robots. There is no means of communication for the robots except an externally visible persistent light on them, which can determine color from a prefixed color set. Since there is a fixed number of colors, this type of communication is considered as a weak form of communication between the robots. Other than this light, robots do not have any persistent memory to store past information (this type of robot is called \emph{oblivious luminous robots}). 

We, in this paper, initiate the study of distributed uniform partitioning of a bounded region using opaque luminous mobile robots. The activation schedule of the robots is \emph{asynchronous} (ASYNC) where any robots can be activated at any time and there is no notion of a global time. The primary motive is to use a swarm of mobile robots with assumptions as less as possible. The problem aims to arrange the robots in such a way that the region gets divided into equal partitions and each partition contains exactly one robot. We retain our focus on this objective only, whereas the problem could even be extended to a version where the robots have sufficient memory or some extra ability to store the coordinates of the partitions. The problem gets challenging due to the oblivious nature of the robots. The obstructed visibility of the robots adds to the challenge even more, because it is not always possible to count the number of robots present in the region. We also assume that the robots do not have any knowledge about the total number of robots. Although mutual visibility algorithms can be used to make any three robots non-collinear and robots can count the total number of robots, the obliviousness of the robots disables them to keep the information beyond the current LCM cycle. From the application point of view, we assume that the robots can detect the boundary of the region. We propose algorithms for the robots considering the region to be a standard geometric shape such as a rectangle, square or circle. Moreover, the algorithms are collision-free which means that two robots can never collocate at the same point. 

\subsection{Contributions}
Our contributions, in this paper, are listed below.
\begin{itemize}
    \item To the best of our knowledge, this is the first work towards the problem of distributed uniform partitioning of a bounded region $\Re$ using a swarm of $N$ ASYNC oblivious opaque mobile robots. 
    \item We propose Algorithm \textsc{Rectangle\_Partition} when $\Re$ is a rectangular region. The algorithm requires $2$ colors which is optimal with respect to the number of colors. The algorithm runs in $O(N)$ epochs, where $N$ is the total number of robots deployed in $\Re$.
    \item We propose Algorithm \textsc{Square\_Partition} when $\Re$ is a square region, which runs in $O(N)$ epochs and uses $5$ colors.
    \item We propose Algorithm \textsc{Circle\_Partition} when $\Re$ is a circle, which runs in $O(N^2)$ epochs and uses $8$ colors. All of our algorithms are collision-free.
\end{itemize}

\subsection{Related works}
Distributed algorithms using mobile robots have been widely studied for many years. Some of the popular topics in this area are gathering, pattern formation, mutual visibility, etc. These problems are studied using different robot models in the literature that provide us with an extensive idea about how mobile robots work. In various literature \cite{FLOCCHINI2008412,10.1145/3569551.3569553,10.1145/3007748.3007781}, the pattern formation problem is studied with mobile robots with different capabilities.
Formation of regular geometric shape can have a great correlation with the problem of uniform partitioning. Uniform circle formation (popularly known as UCF) is one of the most popular problems in this area which can be useful for partitioning a region of a specific geometrical shape. 
In the recent past,  Feletti et al. \cite{app13137991} presented a linear time algorithm for the UCF problem using ASYNC luminous robots.
Later in \cite{feletti_et_al:LIPIcs.OPODIS.2023.5}, they proposed an improved algortihm that runs in $O(\log(N))$ epochs.
Mutual visibility is another popular area where mobile robots are used. To achieve mutual visibility, the movement of the robots needs to be designed carefully enough to avoid collision between robots. 
Several papers \cite{di2017mutual,10183479,sharma2021constant} highlight this important aspect of mobile robots which is an essential part of designing the movement of the robot because collision might lead to a multiplicity points, whereas multiplicity detection is a costlier task. 
As mentioned earlier, the uniform partitioning problem has a wide range of applications such as surveillance of a particular region, painting a bounded region and many more.
Saha et al. \cite{9187545} inspected surveillance of uneven surfaces using drones which falls into the category of coverage problem. Their target is to give a compact coverage of the area so that the diameter of the drone network gets minimized.  
Das and Mukhapadhyaya \cite{Das2013DistributedPB} proposed an asynchronous algorithm for distributed painting in a rectangular region with a robot swarm where robots have agreement on a global line.
The algorithm divides the region into uniform horizontal strips, but one of the advantages in this paper is that the robots are transparent which enables them to see all other robots in the region in every activation. Later in \cite{das2018distributed}, they studied the distributed painting with robots having limited visibility and a global coordinate system. Robots are not completely oblivious and transparent so a robot can see all the robots within its visibility range. Das et al. \cite{das2021swarm} extended the distributed painting when the rectangular region has opaque obstacles in it. They consider the robots to be transparent and work under semi-synchronous settings having total agreement in the direction and orientation of their local coordinate systems.  
Pavone et al. \cite{5710394} proposed an algorithm for equitable partitioning of an environment using synchronous mobile agents. Their algorithm uses the Voronoi-based partitioning approach, but our proposed algorithms can partition a region uniformly using mobile robots working under asynchronous setting in a much simpler way. Acevedo et al. \cite{acevedo2016distributed} gave an algorithm for partitioning a known region using aerial robots, but the partitioning is not uniform.

\section{Model and Preliminaries}\label{model}

\noindent \textbf{Robots:}
A set of $N$ \emph{autonomous}, \emph{anonymous}, \emph{homogeneous} and \emph{disoriented} mobile robots $ \{ r_1, r_2, \cdots r_N \}$ is deployed at distinct points within a bounded region.
This region can be thought of as a subset of the Euclidean plane. Each robot is considered as a point on the plane.  
We consider the robots to be \emph{opaque} because of which a robot $r_i$ is unable to see the robot $r_k$ if there is a robot $r_j$ lying on the line segment joining $r_i$ and $r_k$. However, robots can detect the boundary of the region. A robot is visible to itself but might be unable to see all the robots inside the region due to obstructed visibility. Moreover, robots do not know $N$. Each robot has its local coordinate system, and the current position of the robot is considered the origin of its coordinate system. 

Each robot has a persistent light externally visible to all other visible robots. This light can determine color from a prefixed set of colors. These colors enable the robots to have a weak form of communication among themselves. $r_i.color$ represents the current color of the robot $r_i$ at any time. Other than this persistent light, robots do not have any other memory to remember any information from the past. We misuse the notation $r$ to denote the current position of the robot $r$. Two robots exhibit collision if they are collocated at a point simultaneously. 
\vspace{2mm}

\noindent \textbf{Region:} We are given with a bounded region, denoted by $\Re$. The \emph{interior} of the region $\Re$, denoted by $Int(\Re)$ is defined to be the part of $\Re$ without the boundary. In this paper, we consider the region to be a standard geometric region such as a rectangle, square or circle. A robot lying on $Int(\Re)$, is called an \emph{interior robot}. A robot is a \emph{boundary robot} when it lies on the boundary of $\Re$, but not on the corners. When $\Re$ has corners, robots lying on them are called \emph{corner robots}. The boundary of $\Re$ is identifiable by the robots from any point in the region which enables them to identify whether the region is a rectangle, square or a circle. For any two points $A$ and $B$ in the region, $\overline{AB}$ denotes the line segment joining $A$ and $B$. $\overleftrightarrow{AB}$ is the line passing through the two points. The length of a line segment $\overline{AB}$ is represented by $len(\overline{AB})$. We denote the distance between two points $A$ and $B$ by $d(A, B)$. A similar notation $d(p, L)$ is used to represent the shortest distance between a point $p$ and a line $L$.

\vspace{2mm}
\noindent\textbf{Activation Cycle:}
We consider the mobile robots operating in the \emph{Look-Compute-Move} (LCM) cycles, in which the actions of the robots are divided into three phases.
 \textit{Look:} The robot takes a snapshot of its surroundings, i.e., the vertices within the visibility range and the colors of the robots occupying them.
    \textit{Compute:} The robot runs the algorithm using the snapshot as the input and determines a target vertex or chooses to remain in place. It also changes its current color, if necessary.
  \textit{Move:} The robot moves to the target vertex if needed.

\vspace{2mm}
\noindent \textbf{Scheduler and Run-time:} Robots are activated under a fair \emph{asynchronous} (ASYNC) scheduler. There is no common notion of time in the ASYNC setting. Any number of robots can be activated at any time, with the fairness assumption that each robot is activated infinitely often. Time is measured in terms of \emph{epochs}, which is the smallest time interval in which every robot gets activated and executes its LCM cycle at least once. 

\noindent \textbf{Problem Definition (Distributed Uniform Partitioning)}:  $N$ oblivious, opaque, luminous point robots are deployed at arbitrary distinct points on a bounded region $\Re$. The robots neither have a global agreement on the coordinate axis nor knowledge of $N$. Each of them operates in \emph{Look-Compute-Move} activation cycles and has a persistent light attached to it that can assume a color from a predefined color set. The objective is to divide the region $\Re$ into $N$ uniform partitions using $N$ mobile robots such that each partition contains exactly one robot.

\begin{figure}
    \centering
    \includegraphics[width=\linewidth]{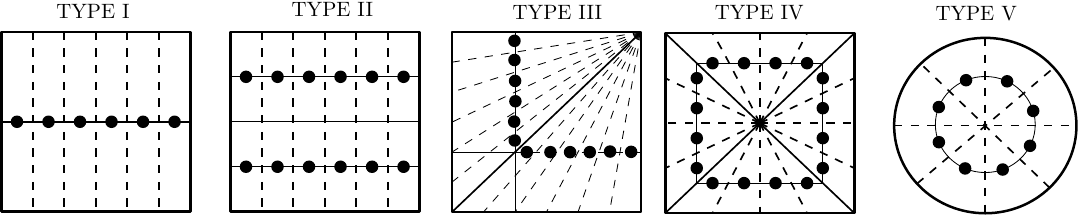}
    \caption{Illustrating different types of partitioning}
    \label{partitions}
\end{figure}
\noindent \textbf{Partition Types:} Depending on the positions and distribution of the robots over the region $\Re$, robots decide the type of partitioning in a distributed manner. All the partitions should have the same area. We identify four types of partitioning for the rectangular and square regions, as shown in Fig. \ref{partitions}. When the region $\Re$ is a rectangle, robots follow either Type I or Type II partitioning, whereas, in the case of a square region, robots terminate in one of the four types (Type I, II, III and IV). If $\mathscr{P}_1, \mathscr{P}_2, \cdots, \mathscr{P}_N$ are the sub-regions of equal area (i.e., partitions of $\Re$), then $\bigcap\limits_{i=1}^N \mathscr{P}_i$ is a point and we call it a \emph{vertex} of Type III and IV partitioning. In the case of a circular region, Type V partitioning is followed.

\begin{table}[!ht]\footnotesize
    \centering
   \caption{The list of colors for different regions with their specification}
\begin{tabular}{|p{1.7cm}|p{1.5cm}|p{9cm}|}\hline
\textbf{Region} & \textbf{Color} & \textbf{Specification} \\ \hline
   \multirow{2}{*}{Rectangle} & \texttt{OFF} & Initial color \\ \cline{2-3}
   
  & \texttt{FINISH} & Used by robots at termination for Type I and II partitions \\ \hline 
   \multirow{5}{*}{Square} & \texttt{OFF} & Initial color \\ \cline{2-3}
   & \texttt{\texttt{MONITOR}} & Used when the robot moves to the apex point to count the number of robots in case of robots are only on one or two opposite sides of $\Re$\\ \cline{2-3}
   
   & \texttt{FINISH} & Used by robots at termination for Type I and II partitions \\ \cline{2-3}
   & \texttt{FINISH1} & Used by robots at termination for Type III partition \\ \cline{2-3}
   & \texttt{FINISH2} & Used by robots at termination for Type IV partition\\ 
    \hline
    \multirow{9}{*}{Circle}& \texttt{OFF} & Initial color of the robots\\
    \cline{2-3}
   & \texttt{HEAD} & Used for the head of  an eligible cluster\\
   \cline{2-3}
   & \texttt{TAIL} & Used for the tail of an eligible clusters\\
    \cline{2-3}
    & \texttt{MID} & Used for the middle robots of an eligible cluster\\
    \cline{2-3}
    & \texttt{MOVE-H} & Used for the movement of the head towards the other cluster when the head moves half of the excess distance\\
    \cline{2-3}
    & \texttt{HALF} & Used after the color \texttt{MOVE-H}\\
    \cline{2-3}
    & \texttt{FULL} & Used for the movement of the head towards the other cluster when the head moves the excess distance completely\\
    \cline{2-3}
    &\texttt{FINISH} & Used by robots at termination\\
    \hline
  
\end{tabular}

\label{color_table}
\end{table}
\noindent \textbf{Organization:} Rest of the paper is organized as follows. Section \ref{rectangle} discusses the algorithm for the rectangular region. Section \ref{square} explores the algorithm when the region is a square. Section \ref{circle} mentions the algorithm for the circular region. Finally, we conclude in Section \ref{conclusion}. 

\section{Algorithm for a Rectangular Region}
\label{rectangle}

This section proposes the algorithm \textsc{Rectangle\_Partition} when the region $\Re$ is a rectangle. The pseudocode of the algorithm can be found in  \ref{appendix}.
Initially, robots are deployed over distinct points on the rectangle $\Re$ with color \texttt{OFF}. 
Robots will follow either Type I or Type II partitioning, which our algorithm ensures. The algorithm, in this case, uses $2$ colors which are described in the Table \ref{color_table}. 
Our strategy is to bring the robots to the rectangle's boundary first. Then, each robot will occupy a point of a partition of the rectangle.
We define some notations that will be used to explain the algorithm. We must note that the variables are defined with respect to the local coordinate system of a robot $r$. A point visible to two different robots can be perceived differently, depending on their coordinate system. 

\noindent \textbf{Notations:} For any robot $r$, we define the following notations.
\begin{itemize}
    \item $S_r$ is the nearest longest side of the rectangle $\Re$ to $r$ or it denotes the specific longest side if $r$ lies along one of them.
   
    \item $S_r^{opp}$ denotes the side opposite to $S_r$ in the rectangle.
    
    \item $L_r$ is the line perpendicular to $S_r$ that passes through $r$.
    
    \item $p_r$ is the point of intersection of the two lines $L_r$ and $S_r$.
    
    \item The two endpoints of a side $S$ of $\Re$ is denoted by $e^1_{S}$ and $e^2_{S}$.
    \item For two opposite sides $S$ and $S'$ of $\Re$ and a constant $0<k<1$, $k S$ is the line segment of length $len(S)$ parallel to the side $S$ that satisfies $d(kS, S) = k d(S, S')$ and intersects $\Re$ at two points.
    
    \item For any two constants $k$ and $k'$ with $0 < k\neq k'<1$ and for a side $S$ of $\Re$, $[kS, k'S]$ is the rectangular region bounded by $kS$, $k'S$ and two adjacent sides of $S$. We denote the interior of the region $[kS, k'S]$ by $(kS, k'S)$. We also denote $(kS, k'S) \cup k'S$ by $(kS, k'S]$.
    
    \item $\mathcal{CH}_r$ is the local convex hull of all the visible robots of $r$.
\end{itemize}

\subsection{Description of the Algorithm}
After activation with $r.color =$ \texttt{OFF}, the target of $r$ is to reach one of the two longest sides of $\Re$ and waits till all other interior robots, all the visible corner robots, and the robots on the two shortest sides with color \texttt{OFF} reach on one of the two longest sides of $\Re$. Before moving to one of the longest sides of $\Re$, $r$ needs to choose a side of $\Re$ as its target. Let us denote this target side by $S_{r}$. If $r$ is already situated on one of the two longest sides of $\Re$, then it selects that side as $S_r$.
 If $r$ is a corner robot, $r$ chooses the longest side incident with it as $S_r$. If $r$ is an interior robot or a boundary robot on one of the two shortest sides, it chooses the nearest longest side as $S_r$. In case of $r$ being equidistant from both the longest sides, $r$ chooses any one of them as $S_r$.

\noindent \textbf{Strategy} \textsc{Move\_To\_LongestSides:} When $r$ gets activated with $r.color =$ \texttt{OFF}, $r$ finds the target side $S_r$ in the current LCM cycle and calculates the point of intersection $p_r$ of $S_r$ and $L_r$. 
Now, $r$ needs to consider its current position in the region provided the point $p_r$ is visible to $r$. 
It 
\begin{wrapfigure}[15]{r}{0.4\textwidth}
\centering
  \includegraphics[width=0.9\linewidth]{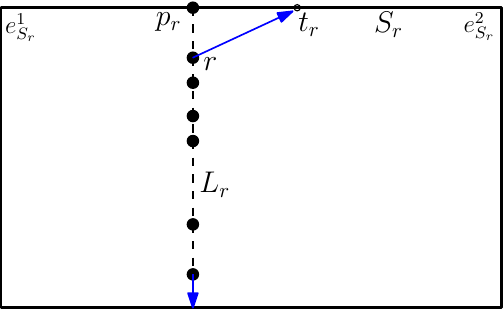}
     \caption{The movement of $r$ to $t_r$ when all robots are on a line}
     \label{fig.v_r_empty}
  \smallskip\par
  \includegraphics[width=0.9\linewidth]{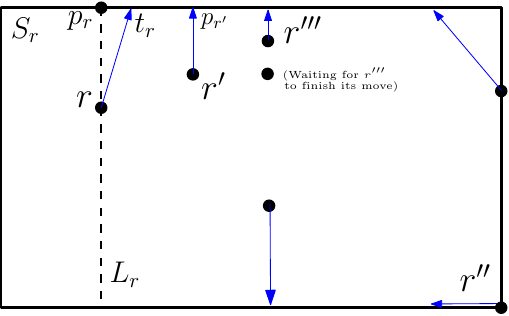}
     \caption{$r$ moves to the point $t_r$ that lies on the nearest longest side $S_r$}
     \label{fig.moving_to_target_side}
\end{wrapfigure}
maintains the status quo if the point $p_r$ is not visible to $r$. 
When $p_r$ is visible, we can differentiate the following four cases based on the different positions of $r$ and propose different strategies for $r$. 
(i) If $r$ is an interior robot with $p_r$ empty, $r$ simply moves to $p_r$ with its current color \texttt{OFF}.
(ii) If $r$ is a boundary robot on one of the two shortest sides of $\Re$ and $p_r$ (which, in this case, is one of the corners of $\Re$) has a robot on it, $r$ waits with no change in its current color till $p_r$ gets empty.
(iii) $r$ is a boundary robot on one of the two shortest sides of $\Re$, and $p_r$ is empty.
(iv) $r$ is either an interior robot with another robot on $p_r$ or a corner robot.
In the last two cases, (iii) and (iv), $r$ needs to find a target point $t_r$ on $S_r$ in such a way that it does not encounter any collision. 
 In this process, $r$ finds a set $\mathcal{V}_r$ that consists of all visible robots not lying on $L_r$. There can be two sub-cases. 
\begin{itemize}
    \item \textbf{$\mathcal{V}_r$ is empty:} In this case, $r$ selects a target point $t_r$ on the side $S_r$ as shown in Fig. \ref{fig.v_r_empty}, such that $d(p_r, t_r) = \frac{1}{2} \max \{ d(e^1_{S_r},p_r), d(e^2_{S_r}, p_r) \}$.

    \item \textbf{$\mathcal{V}_r$ is non-empty:} As shown in Fig. \ref{fig.moving_to_target_side}, $r$ in this case, chooses the target point $t_r$ that satisfies $d(p_r,t_r) = \frac{1}{4} \min\limits_{r' \in \mathcal{V}_r} \{ d(r',L_r) \}$.
\end{itemize}
Finally, $r$ moves to the point $t_r$ without changing the current color \texttt{OFF}.  If $r$ is already positioned on one of the longest sides of $\Re$, it does not change its current position or color until all the \texttt{OFF}-colored robots not lying on the longest sides reach on their nearest longest side. 

Now, we need the following definitions to categorize the robots.
\begin{definition}{(Terminal Robot):} Let $L$ be a line segment with two endpoints $e_{L}^1$ and $e_{L}^2$. A non-corner robot $r$ lying on $L$ is called a terminal robot on $L$ if at least one of line segments $\overline{re_L^1}$ and $\overline{re_L^2}$ is not occupied by any other robot. 
\end{definition}

\begin{definition}{(Monitor Robot):}
\label{monitorrobot}
Let $S_r$ be the longest side of $\Re$ where $r$ is situated in the current LCM cycle. $r$ is called a monitor robot if it satisfies all the following conditions. 
    (i) $r$ is a terminal robot on $S_r$.
    (ii) $(S_r, \frac{7}{8}S_r)$ has no robot and all the robots on $(S_r^{opp}, \frac{1}{8}S_r^{opp}]$ are with color \texttt{FINISH}.
    (iii) There is no corner robot visible to $r$.
    (iv) There is no robot on both the shortest sides of $\Re$.
\end{definition}

After all the visible robots lying either on $S_r$ or on $S_r^{opp}$ (where $S_r$ is one of the longest sides of $\Re$), the robot $r$ finds whether it is a monitor robot or not. The strategy is to move the monitor robots at a particular distance (inside the region) from the longest sides of $\Re$ so that they can count the number of total robots present in the region and decide the partition type of $\Re$. This is important because neither the robots possess any memory to store the value of $N$ nor are they transparent to see each other. It does nothing if $r$ does not qualify as a monitor robot. When $r$ is a monitor robot on $S_r$, it needs to move to a point $a_r$ on $\frac{1}{8}S_r$, referred to as \emph{apex point}, from where it can count the number of robots which is necessary to decide the next action of $r$. The robot $r$ moves to its apex point with color \texttt{FINISH}. At this stage, our target is to gather all the robots on one of the longest sides of $\Re$ if the number of robots lying on the two longest sides differs. In this case, our objective is to relocate all the robots from the side with a smaller number to the side with a larger number of robots. After all the robots have positioned themselves on one of the longer sides of $\Re$, monitoring robots place themselves again at apex points with the color \texttt{FINISH} to count $N$ and then move to their final positions of the partitioning on $\frac{1}{2}S_r$ without changing their color. Other robots decide their final positions based on these \texttt{FINISH}-colored robots on $\frac{1}{2}S_r$. We refer to this kind of partitioning as Type I partitioning.

On the other hand, if the number of robots on both of the longest sides is the same, we adopt a different strategy to partition the region. Here, the robots on the apex points first place themselves in their respective final positions on $\frac{1}{4}S_r$ without changing their color. Other robots decide their final positions based on the \texttt{FINISH}-colored robots. We refer to this as Type II partitioning.

In the following, we describe in detail the action of a monitor robot $r$ lying on $S_r$.

\noindent\textbf{Strategy} \textsc{MonitorRobot\_Movement\_Rectangle}:
 If each of $S_r$ and $S_r^{opp}$ has exactly one robot on it, $r$ simply moves to the midpoint of the line segment $\frac{1}{4}S_r$ from $S_r$, after changing its color to \texttt{FINISH}. If $S_r^{opp}$ has exactly one robot, but $S_r$ has multiple robots on it, $r$ can understand this scenario by confirming the existence of its neighbouring robot on $S_r$. In this case, $r$ waits until the robot on $S_r^{opp}$ moves to $S_r$. If $r$ is the only robot on $S_r$, but $S_r^{opp}$ has multiple robots, $r$ decides its target side as $S_r^{opp}$ and moves to a point on it with the color \texttt{OFF}, following a strategy similar to the strategy \textsc{Move\_To\_LongestSides}. Observe that in all of the above scenarios, $r$ does not need to compute and move to its apex point on $\frac{1}{8}S_r$. We now focus on the case where both $S_r$ and $S_r^{opp}$ have multiple robots, and the monitor robot $r$ moves to the apex point $a_r$.
\begin{figure}[!ht]
 \begin{minipage}[c]{0.32\textwidth}
     \centering
     \includegraphics[width=0.95\linewidth]{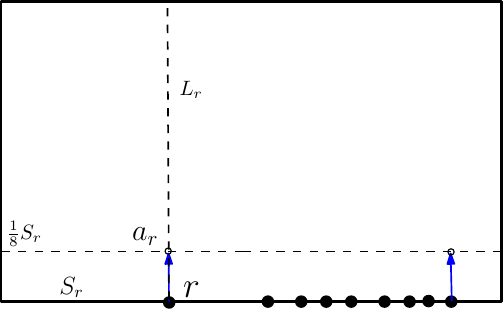}
     \caption{If $S_r^{opp}$ has no robot, $a_r$ is the intersection point of $L_r$ and $\frac{1}{8}S_r$}
     \label{fig.apex_point_s_r_opposite_empty}
 \end{minipage}
 \hfill
 \begin{minipage}[c]{0.32\textwidth}
     \centering
     \includegraphics[width=0.95\linewidth]{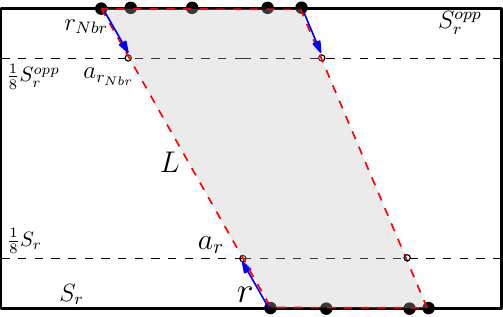}
     \caption{$r$ moves to $a_r$ considering $r_{Nbr}$ on $\mathcal{CH}_r$, which is not  on $S_r$}
     \label{fig.move_apex_point}
 \end{minipage}
\hfill
\begin{minipage}[c]{0.32\textwidth}
     \centering
     \includegraphics[width=0.95\linewidth]{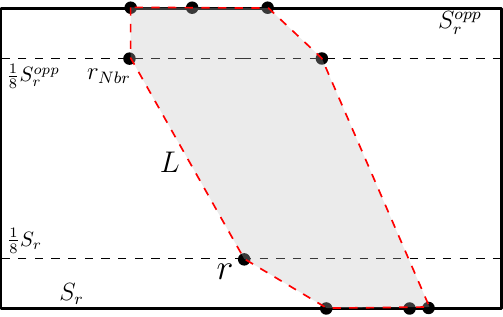}
     \caption{$r$ remains a vertex of  $\mathcal{CH}_r$ after the movement to apex point}
     \label{fig.after_movement_to_apexpoint}
 \end{minipage}
 \end{figure}

  \noindent \textbf{Movement to the Apex Point $a_r$:}
  When $r$ finds no robot on $S_r^{opp}$, it selects the point of intersection of $L_r$ and $\frac{1}{8}S_r$ as its apex point $a_r$. as shown in Fig \ref{fig.apex_point_s_r_opposite_empty}. When there is at least one robot on $S_r^{opp}$, $r$ first selects the line $L = \overleftrightarrow{rr_{Nbr}}$, where $r_{Nbr}$ is the neighbour of $r$ on the convex hull $\mathcal{CH}_r$ lying on the side $(S_r^{opp}, \frac{1}{8}S_r^{opp}]$. Then, $r$ chooses the intersection point of $L$ and $\frac{1}{8}S_r$ as the apex point $a_r$ (refer to Fig \ref{fig.move_apex_point}). 
  The calculation of the apex point depends on the positions of the robots on $S_r^{opp}$ to confirm that $r$ can see all other robots in $\Re$ and does not become an obstruction between the other two robots after its movement, as depicted in Fig \ref{fig.after_movement_to_apexpoint}.
  Finally, $r$ moves to $a_r$ after changing its color from \texttt{OFF} to \texttt{FINISH}.

We now identify two cases based on the position of $r$. 

\noindent \textbf{Case 1 ($r$ lies on $\frac{1}{8}S_r$):} In this case, if $r$ gets activated and finds itself lying on $\frac{1}{8}S_r$, then $r.color$ must be \texttt{FINISH}, as $r$ must have moved to its apex point from the side $S_r$ in its previous LCM cycle. Now, $r$ identifies the nearest longest side of $\Re$ as $S_r$. 
It is possible for $r$ to encounter an \texttt{OFF}-colored robot $r'$ on $Int(\Re)$, which might occur if $r'$ is in its move phase to allocate itself on 
\begin{wrapfigure}[9]{r}{0.4\textwidth}
\centering
  \includegraphics[width=0.9\linewidth]{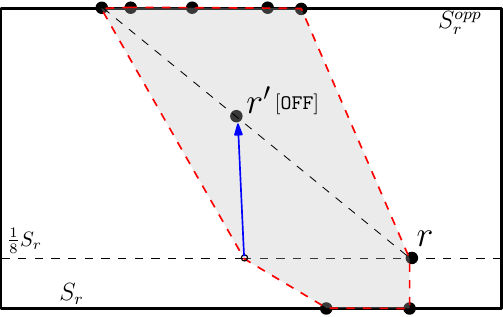}
     \caption{Movement of the \texttt{OFF}-colored robot $r'$ creates obstruction for $r$}
     \label{fig.off_colored_robot_creates_obstruction}
\end{wrapfigure}one of the longest sides of $\Re$ (either from $\frac{1}{8}S_r$ or from $\frac{1}{8}S_r^{opp}$). Fig. \ref{fig.off_colored_robot_creates_obstruction} depicts one such situation. This might occur due to the asynchronous activation of the robots. 
Here $r$ does not change its position and color.
Otherwise (when no \texttt{OFF}-colored robot is visible in $Int(\Re)$), $r$ can identify the other \texttt{FINISH}-colored robot $r_i$ within $[S_r, \frac{1}{8}S_r]$, if exists (which might have been coming from $S_r$ to its apex point). Moreover, it can also identify the robot with color \texttt{FINISH} on $[\frac{1}{8}S_r, \frac{1}{2}S_r]$ (this is possible when $r_i$ is moving to its final position either on $\frac{1}{4}S_r$ for Type II partitioning or on $\frac{1}{4}S_r$ for Type I partitioning and it is observed by $r$).
  Similarly, $r$ can easily identify the robots with color \texttt{FINISH} which are nearest to $S_r^{opp}$ by observing the region $[S_r^{opp}, \frac{1}{4}S_r^{opp}]$.
  If $r$ finds all the \texttt{FINISH}-colored robots on $\frac{1}{8}S_r$ and $\frac{1}{8}S_r^{opp}$, it calculates $c_r^1$ and $c_r^2$ where $c_r^1$ is the total number of the robots on the band $[S_r,\frac{1}{8}S_r]$ including itself. $c_r^2$ is the total number of robots on the band $[S_r^{opp}, \frac{1}{8}S_r^{opp}]$.  
  We now identify the following four sub-cases for the robot $r$ lying on $\frac{1}{8}S_r$ based on  $c_r^1$ and $c_r^2$. 

  \begin{figure}[h]
 \begin{minipage}[c]{0.47\textwidth}
     \centering
     \includegraphics[width=0.9\linewidth]{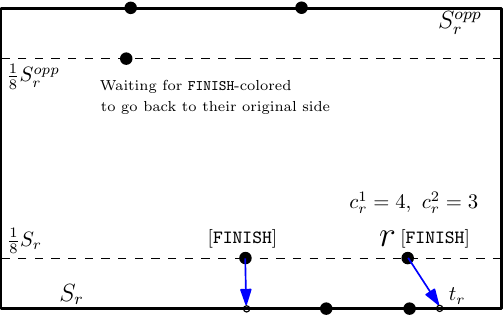}
     \caption{$r$ moves back to $S_r$ from its apex point when $c_r^1 > c_r^2 >0$}
     \label{fig.case1.1_rectangle}
 \end{minipage}
 \hfill
 \begin{minipage}[c]{0.47\textwidth}
     \centering
     \includegraphics[width=0.9\linewidth]{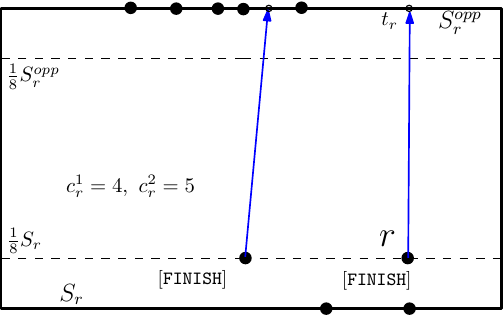}
     \caption{$r$ moves to $t_r$ on $S_r^{opp}$ when $0< c_r^1<c_r^2$ with no \texttt{FINISH}-colored robot on $[S_r^{opp}, \frac{1}{8}S_r^{opp}]$}
     \label{fig.case1.2_rectangle}
 \end{minipage}
 \end{figure}


\begin{itemize}
	 \item \textbf{Case 1.1 ($0 < c_r^2 < c_r^1$):} \label{subhajit} It means that there is a lesser number of robots near $S_r^{opp}$ than that of $S_r$, as shown in Fig. \ref{fig.case1.1_rectangle}. In this case, our aim is to relocate all the robots lying on $[S_r^{opp}, \frac{1}{4}S_r^{opp}]$ to $S_r$. $r$ also moves back to $S_r$ in this case. If the point of intersection $p_r$ of $L_r$ and $S_r$ is visible to $r$ and contains no robot, $r$ moves to $p_r$ with $r.color =$ \texttt{OFF}. Otherwise, it chooses a target point $t_r$ on $S_r$ such that $d(p_r, t_r) = \frac{1}{4} \min\limits_{r' \in \mathcal{V}_r} {d(r', L_r)}$, where $\mathcal{V}_r$ is the set of all visible robots not lying on the line $L_r$. Then, $r$  moves to $t_r$ with current color \texttt{FINISH}.
	 
      \item \textbf{Case 1.2 ($0 < c_r^1 < c_r^2$):} It means that there are fewer robots near $S_r$ than that of $S_r^{opp}$ (Fig \ref{fig.case1.2_rectangle}). In this case, $r$ needs to move to $S_r^{opp}$. Before the movement to $S_r^{opp}$, it first checks whether there is any \texttt{FINISH}-colored robot in $(S_r^{opp}, \frac{1}{8}S_r^{opp}]$. If such a robot $r_1$ exists, then in ASYNC settings, the movement of $r_1$ might lead to a miscount of the number of the robots for $r$, as shown in Fig \ref{fig.wrong_calculation_of_r}. In this situation, $r$ waits until these robots allocate themselves on $S_r^{opp}$. Whenever $r$ does not find any \texttt{FINISH}-colored robots in $(S_r^{opp}, \frac{1}{4}S_r^{opp}]$, it decides to move on $S_r^{opp}$ by finding the target point $t_r$ on $S_r^{opp}$. It needs to consider $p_r^{opp}$, which is the point of intersection of $L_r$ and $S_r^{opp}$. Using the same strategy described in Case 1.1 (instead of $S_r$, $r$ considers $S_r^{opp}$), $r$ finally moves on $t_r$ after changing its current color to \texttt{OFF}. 
\begin{figure}[H]
 \begin{minipage}[c]{0.47\textwidth}
     \centering
     \includegraphics[width=0.9\linewidth]{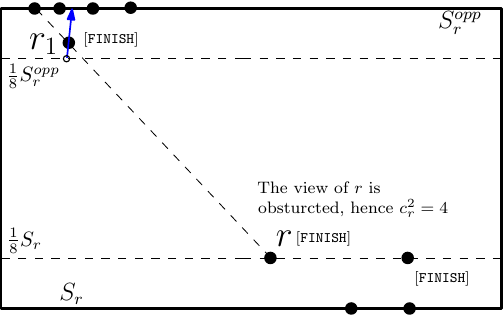}
     \caption{$r$ miscalculated $c_r^2$ if a robot lies on $[S_r^{opp}, \frac{1}{8}S_r^{opp}]$ moving to its target point on $S_r^{opp}$}
     \label{fig.wrong_calculation_of_r}
 \end{minipage}
 \hfill
 \begin{minipage}[c]{0.47\textwidth}
     \centering
     \includegraphics[width=0.9\linewidth]{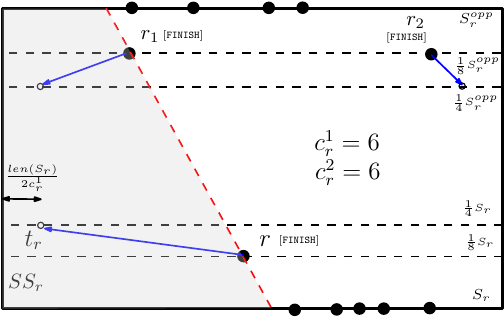}
     \caption{$r$ moves to $t_r$ on $\frac{1}{4}S_r$ because of same number of robot around $S_r$ and $S_r^{opp}$}
     \label{fig.c1=c2}
 \end{minipage}
 \end{figure}

      \item \textbf{Case 1.3 ($0 < c_r^1 = c_r^2$):} In this case, we aim to make the robots follow the Type II partitioning. $r$ first identifies the shortest side of $\Re$, denoted as $SS_r$, which either lies on or intersects the half plane delimited by the line $\overleftrightarrow{rr_1}$ where no other robots of $\Re$ are present. Here $r_1$ is the neighbour of $r$ on $\mathcal{CH}_r$ whose nearest longest side is $S_r^{opp}$. Now, it computes a target point $t_r$ on $\frac{1}{4}S_r$ such that $d(t_r, SS_r) = \frac{len(S_r)}{2c_r^1}$, as shown in Fig. \ref{fig.c1=c2}. Finally, $r$ moves to $t_r$ from its current position with the current color \texttt{FINISH}.
       
      \item \textbf{Case 1.4 ($0 = c_r^2 < c_r^1$):} This case arises when there is no robot lying on $S_r^{opp}$ and the robot $r$ has moved from $S_r$ to a point on $\frac{1}{8}S_r$, as shown in Fig \ref{fig.case1.4_rectangle}. In this case, we want the robots to follow the Type I partitioning. So, $r$ first identifies $SS_r$, which lies on the half plane delimited by $L_r$ where no other robot in $\Re$ resides. It then calculates a target point $t_r$ on $\frac{1}{2}S_r$ such that $d(t_r, SS_r) = \frac{len(S_r)}{2c_r^1}$. It then moves to $t_r$ with the current color \texttt{FINISH}.
  \end{itemize}

    If $r$ is activated on $\frac{1}{8}S_r$ with its color \texttt{FINISH} and sees one \texttt{FINISH}-colored robot $r'$ on $(\frac{1}{8}S_r, \frac{7}{8}S_r)$, it moves back to $S_r$ with color \texttt{FINISH} by following similar movement strategy as described in Case 1.1. This is because the position of $r'$ (which is in its move phase towards the final point) might mislead to the wrong calculation of $c_r^1$ and $c_r^2$. 
    If $r$ sees a \texttt{FINISH}-colored robot on $ [S_r, \frac{1}{8}S_r) \cup [S_r^{opp}, \frac{1}{8}S_r^{opp})$, it does nothing.



\begin{figure}[H]
 \begin{minipage}[c]{0.47\textwidth}
    
     \centering
     \includegraphics[width=0.9\linewidth]{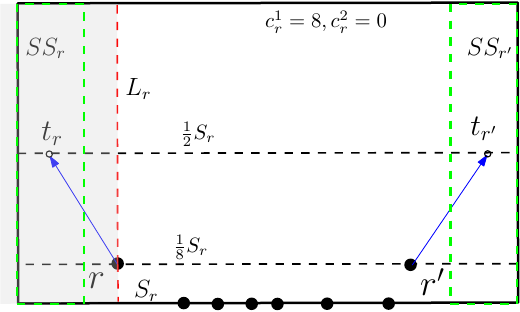}
     \caption{$r$ moves to $\frac{1}{2}S_r$ from its apex point on $\frac{1}{8}S_r$ by following the Type I partitioning. The green box shows a partition of the region}
     \label{fig.case1.4_rectangle}
 \end{minipage}
 \hfill
 \begin{minipage}[c]{0.47\textwidth}
     \centering
     \includegraphics[width=0.9\linewidth]{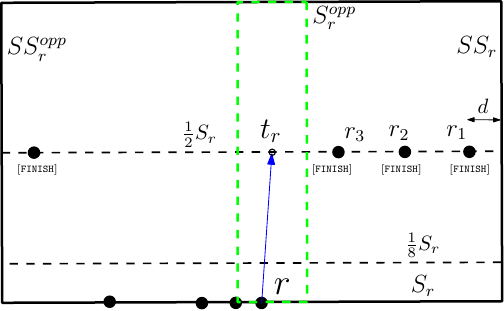}
     \caption{$r$ is on $S_r$ with color \texttt{OFF} and utilizes the position of the \texttt{FINISH} colored robots on $\frac{1}{2}S_r$ to reach its final position}
     \label{fig.case2.3rectangle}
 \end{minipage}
 \end{figure}

  \noindent \textbf{Case 2 ($r$ lies on $S_r$):}  If $r$ gets activated with color \texttt{FINISH} on $S_r$, it understands that it has moved back to $S_r$ from its apex point in its previous LCM cycle. So, it changes its color to \texttt{OFF}. The positions of the \texttt{FINISH}-colored robots in $Int(\Re)$ help the \texttt{OFF}-colored robots on $S_r$ to decide their final positions inside $\Re$. If $r$ is a non-terminal robot on $S_r$, $r$ waits till it becomes terminal on $S_r$. If $r$ becomes a terminal robot with color \texttt{OFF} and all the robots on $Int(\Re)$ are with color \texttt{FINISH}, we list down the following three sub-cases based on the different positions of the \texttt{FINISH}-colored robots, where $r$ decides to move to a target point $t_r$.
  \begin{itemize}

      \item Case 2.1: $r$ finds a \texttt{FINISH}-colored robot on $\frac{1}{4}S_r$ and some robots on $[S_r^{opp}, \frac{1}{4}S_r^{opp}]$)

      \item Case 2.2: $r$ finds a \texttt{FINISH}-colored robot on $\frac{1}{4}S_r^{opp}$, but not on $[S_r, \frac{1}{4}S_r]$ 

      \item Case 2.3: $r$ finds a \texttt{FINISH}-colored robot on $\frac{1}{2}S_r$
  \end{itemize}

In the above cases, $r$ first identifies the sides $SS_r$ and $SS_r^{opp}$. $SS_r$ is the shortest side of $\Re$ for which the intersection point of $S_r$ and $SS_r$ is visible to $r$. $SS_r^{opp}$ is the opposite side of $SS_r$. 
Afterwards, $r$ sets $L = \frac{1}{2}S_r$ when Case 2.3 occurs, as depicted in Fig \ref{fig.case2.3rectangle}. 
Otherwise (for Case 2.1 and Case 2.2), it sets $L = \frac{1}{4}S_r$, as shown in Fig. \ref{fig.case2.2_rectangle}. 
$r$ now considers a set $\mathcal{F}_r$, that consists of all the \texttt{FINISH}-colored robots lying on $\frac{1}{4}S_r$, $\frac{1}{4}S_r^{opp}$, and $\frac{1}{2}S_r$. It now calculates
\begin{wrapfigure}[9]{r}{0.4\textwidth}
\centering
  \includegraphics[width=\linewidth]{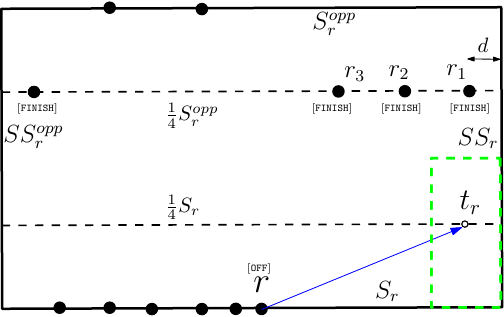}
     \caption{$r$ moves to $\frac{1}{4}S_r$ from $S_r$ by taking help of the \texttt{FINISH}-colored robots on $\frac{1}{4}S_r^{opp}$}
     \label{fig.case2.2_rectangle}
\end{wrapfigure}
$d = \min\limits_{r' \in \mathcal{F}_r} \{ \min \{d(r', SS_r), d(r', SS_r^{opp})\} \}$. Let $r_1, r_2,...r_k$ be a sequence of robots of maximum length on $L$ starting from the robot $r_1$ with $d(r_1, SS_r) =d$ till the robot $r_k$ such that two consecutive robots in the sequence are exactly $2d$ distance apart from each other. When no robot is on $L$, we consider $k = 0$.  Finally, $r$ moves to $t_r$ on $L$ such that $d(t_r, SS_r) = (2k+1)d$ after changing its color to \texttt{FINISH} from \texttt{OFF}. 
In all other situations, it is possible that a \texttt{FINISH}-colored robot is in its move phase towards its target point or it is on its apex point. In those cases, $r$ maintains the status quo.
 
 When $r$  gets activated with the color \texttt{FINISH} and finds itself on either  $\frac{1}{4}S_r$ or  $\frac{1}{2}S_r$, it terminates.

\subsection{Analysis of the Algorithm}
\label{analysis_rectangle}
Here, we analyse the algorithm \textsc{Rectangle\_Partition}. We prove the correctness and the time complexity of this algorithm. We also show that the movement of the robots is free from collision. For our convenience, we refer to the rectangular region $\Re$ as $ABCD$ with two longest sides $\overline{AB}$ and $\overline{CD}$ throughout the analysis of \textsc{Rectangle\_Partition} algorithm.

\begin{lemma}\label{lemma3.1}
All the robots not lying on one of the two longest sides of $\Re$, move to one of the longest sides of $\Re$ without any collision. 
    
\end{lemma}

\begin{proof}
Let $\overline{EF}$ divide the rectangle into two halves, as shown in Fig. \ref{fig.rectangle_lemma3.1}. The robots lying inside the rectangle $ABFE$ choose the side $\overline{AB}$ as their target and move to some points on it. Similarly, the robots inside the rectangle $CDEF$ move to some point of the side $\overline{CD}$. The robots on $\overline{EF}$ choose any of the two sides $\overline{AB}$ and $\overline{CD}$ as the target side and move to it. Let us consider a robot $r$ on $\overline{EF}$ that chooses the side $\overline{AB}$ as the target side, but there are robots on the line segment $\overline{rp_r}$ which are nearer to $\overline{AB}$. In this case, $r$ waits until all these robots reach some point on $\overline{AB}$. So, eventually, $r$ gets the chance to move to a point on $\overline{AB}$. 

Now we choose another robot $r'$ lying inside $ABFE$, but not on $L_r$. We show that $r$ and $r'$ do not collide in their way to $AB$. For both $r$ and $r'$, if the two points $p_r$ and $p_{r'}$ remain empty, they follow the path $\overline{rp_r}$ and $\overline{r'p_{r'}}$ to move to $\overline{AB}$. Since, $\overline{rp_r} || \overline{r'p_{r'}}$, $r$ and $r'$ cannot collide in this movement. Let us assume that $p_r$ and $p_{r'}$ are non-empty and $r'$ is the nearest robot in $\mathcal{V}_r$. So, $r$ calculates a point $t_r$ on $\overline{AB}$ such that $d(p_r, t_r) \leq \frac{1}{4} d(r',L_r) < \frac{1}{2} d(r', L_r)$. Even if $r'$ calculates its target point $t_{r'}$ on the line segment $\overline{AB}$, we have $d(p_{r'}, t_{r'}) \leq \frac{1}{4} d(r',L_r)$. Thus, $d(t_r, t_{r'}) \geq d(p_r, p_{r'}) - d(p_r, t_r) - d(p_{r'}, t_{r'}) \geq d(r', L_r) - \frac{1}{4}d(r', L_r) - \frac{1}{4}d(r', L_r) > 0$. Moreover, the line $L$ passing through the midpoint of $\overline{t_rt_{r'}}$ separates the two paths $\overline{rt_r}$ and $\overline{r't_{r'}}$, as depicted in Fig \ref{fig.rectangle_lemma3.1}.
So, the two robots $r$ and $r'$ cannot collide. In the case of $r$ and $r'$ lying on the same line, they move sequentially to the side $\overline{AB}$, because of which they cannot meet a collision.  
\end{proof}

\begin{figure}[H]
 \begin{minipage}[c]{0.47\textwidth}
      \centering
    \includegraphics[width=0.8\linewidth]{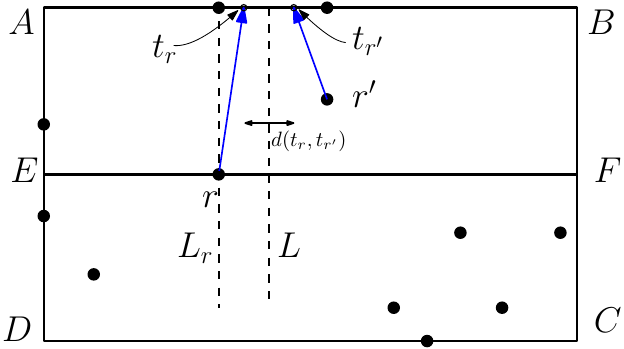}
    \caption{Collision free movement of the robots $r$ and $r'$, even in simultaneous execution}
    \label{fig.rectangle_lemma3.1}
 \end{minipage}
 \hfill
 \begin{minipage}[c]{0.47\textwidth}
     \centering
     \includegraphics[width=0.7\linewidth]{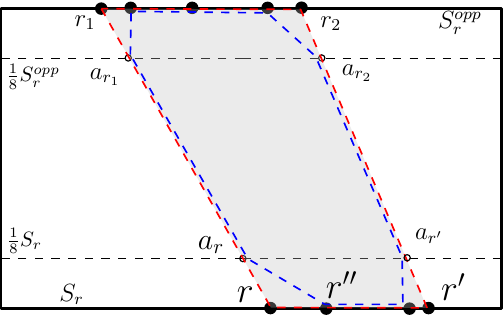}
     \caption{Montior robots move to their respective apex points to see all the robots}
     \label{fig.rectangle_lemma3.2}
 \end{minipage}
 \end{figure}

\begin{lemma}\label{lemma 3.2}
    A monitor robot $r$ lying on $S_r$ selects its apex point $a_r$ on $\frac{1}{8}S_r$ so that all the robots on $\Re$ are visible to it after its movement.
\end{lemma}
\begin{proof}
    We want to take advantage of the properties of the convex hull. Let us consider a convex hull $\mathcal{CH}^G$ of $N$ points where all the points lie on the boundary of $\mathcal{CH}^G$. If a point $v$  and its two neighbours are three vertices of $\mathcal{CH}^G$, then no two points on $\mathcal{CH}^G$ are collinear with $v$.
    By the definition of a monitor robot, $r$ must be a terminal robot on $S_r$. Let us denote the convex hull of all the robots of $\Re$ by $\mathcal{CH}^G$. Let us further assume that $r'$, $r_1$ and $r_2$ are the other monitor robots in $\Re$, where $r'$ lies on $S_r$ and both $r_1$ and $r_2$ lie on $S_r^{opp}$, (Fig. \ref{fig.rectangle_lemma3.2}). 
    Observe that all of $r$, $r'$, $r_1$ and $r_2$ are vertices of $\mathcal{CH}^G$ (shaded region with red boundary), and $r_1$ is one of the neighbours of $r$. After activation, $r$ chooses the point of intersection of the line $\overleftrightarrow{rr_1}$ and $\frac{1}{8}S_r$ as its apex point $a_r$. 
    Since $r$ follows a path to $a_r$ on the line $\overleftrightarrow{rr_1}$, it remains a vertex of $\mathcal{CH}^G$ after the movement. If the other neighbour (lying on $S_r$) of $r$ is $r''$, then $r''$ becomes a vertex of the convex hull $\mathcal{CH}^G$ as $r$ has moved to $a_r$. If the robot $r_1$ stays on $S_r^{opp}$, it remains a vertex of $\mathcal{CH}^G$. In simultaneous movement with $r$, $r_1$ also follows a path on $\overleftrightarrow{rr_1}$, leading $r_1$ to remain a vertex of $\mathcal{CH}^G$. Similarly, if $r'$ and $r_2$ move to their respective apex points, they remain as vertices of $\mathcal{CH}^G$ (shaded region with blue boundary).  So, no two robots become collinear with $r$ after moving to its apex point. 

    When all the robots lie on one side $S_r$ ($r_1$ and $r_2$ do not exist), $a_r$ is the point of intersection of $\frac{1}{8}S_r$ and $L_r$. If another monitor robot $r'$ moves simultaneously with $r$ to its apex point $a_{r'}$ on $\frac{1}{8}S_r$, all other robots lie between $L_r$ and $L_{r'}$ and on $S_r$. Hence, $r$ can see all the robots on $S_r$ after moving to the apex point $a_r$. 
\end{proof}

\begin{lemma}\label{newlemma1_rectangle}
A robot $r$ moves to its apex point $a_r$ on $\frac{1}{8}S_r$ without any collision.
\end{lemma}

\begin{proof}
    When all the robots are positioned on one longest side, say $CD$ of $\Re$, two monitor robots lie on $CD$. In this case, if $r$ and $r'$ be the monitor robots, they choose their respective apex points $a_r$ and $a_{r'}$ on $\frac{1}{8}CD$ (as it is the same as $\frac{1}{8}S_r$ and $\frac{1}{8}S_{r'}$) such that $\overline{r a_r} \perp CD$ and $\overline{r' a_{r'}} \perp CD$. This implies $\overline{r a_r} \parallel \overline{r' a_{r'}}$, which means that the two robots cannot collide. 

    When both the longest sides have some robots, there can be at most four monitor robots. Without loss of generality, we assume that $r_1$ and $r_2$ are the two monitor robots on $S_r^{opp}$ along with $r$ and $r'$ being the monitor robots on $S_r$. As depicted in Fig \ref{fig.rectangle_lemma3.2}, $r_1$ is a neighbour of $r$ on the convex hull $\mathcal{CH}^G$. $r$ cannot meet a collision with $r_1$ and $r_2$ even in case of simultaneous movement because $a_r$ lies on $\frac{1}{8}S_r$ whereas $a_{r_1}$ and $a_{r_2}$ lie on $\frac{1}{8}S_r^{opp}$ (which is same as $\frac{7}{8}S_r$). It is also not possible for $r$ to collide with $r'$, as $r$ and $r'$ follow the paths $\overline{rr_1}$ and $\overline{r'r_2}$ which are two different sides of the convex hull $\mathcal{CH}^G$, as described in Lemma \ref{lemma 3.2}.
\end{proof}

\begin{remark}\label{remark_1_rectangle}
    When neither there is a \texttt{OFF}-colored robot in $Int(\Re)$ nor a \texttt{FINISH}-colored robot on $(\frac{1}{8}S_r, \frac{7}{8}S
    _r)$, a \texttt{FINISH}-colored robot lying on $\frac{1}{8}S_r$ accurately calculates $c_r^1$ and $c_r^2$,  where $c_r^1$ is the total number of the robots on $S_r$ and on the band $[S_r,\frac{1}{2}S_r]$ including itself. $c_r^2$ is the total number of robots on  $S_r^{opp}$ and on the band $[S_r^{opp}, \frac{1}{4}S_r^{opp}]$.
\end{remark}
    
\begin{lemma}\label{lemma3.3}
    If the longest side $AB$ of $\Re$ contains less number of robots than the other longest side $CD$, all the robots on $AB$ move to $CD$ without collision.
\end{lemma}

\begin{proof}
    Let $r$ be a robot on $AB$ which means $S_r = AB$ and $CD = S_r^{opp}$. Let us assume that $r$ is a terminal robot on $S_r$.  If there exists a \texttt{OFF}-colored robot $r'$ in $Int(\Re)$, either $r'$ is moving towards one of the longest sides from its apex point, or it is still in its initial position. In both situations, $r'$ eventually allocates itself on one of the longest sides of $\Re$, by Lemma \ref{lemma3.1}. So $r$ eventually becomes a monitor robot. It then moves to its apex point $a_r$ and checks whether there is any \texttt{FINISH}-colored robot $r''$ in $[S_r^{opp}, \frac{1}{8}S_r^{opp}]$. When $r''$ finds $c^1_{r''} > c^2_{r''}$, it moves back to $S_r^{opp}$. After this, $r$ accurately computes $c^1_r$ and $c^2_r$ by Remark \ref{remark_1_rectangle}. By the assumption of this Lemma \ref{lemma3.3}, $c^1_r < c_r^2$, so it finds $p_r^{opp}$.
 $r$ moves to $p_r^{opp}$, if it is empty. If not, $r$ finds a point $t_r$ on $S_r^{opp}$ such that $d(t_r,p_r^{opp}) = 1/4 \min\limits_{r' \in \mathcal{V}_r} d(r', L_r)$ and moves to the point $t_r$. By similar arguments presented in Lemma \ref{lemma3.1}, this type of movement is collision-free even if there is another robot simultaneously moving towards $S_r^{opp}$. If $r$ is not a terminal robot on $S_r$, it eventually becomes a terminal robot and thus, every robot on $AB$ moves to some point on $CD$.  
\end{proof}
\begin{remark}
    \label{eventually_equal_or_empty}
    Eventually, either one of the longest sides contains all of the robots of $\Re$, or both the longest sides contain an equal number of robots.
\end{remark}

\begin{lemma}
    \label{apex_to_termination}
    For a \texttt{FINISH}-colored robot $r$ lying on $\frac{1}{8}S_r$ with $0 < c^1_r = c^2_r$ or $0 = c^2_r < c^1_r$, when there is no robot on $(\frac{1}{8}S_r, \frac{7}{8}S_r)$, it terminates to a unique partition of $\Re$ and the movement is collision-free.
\end{lemma}
\begin{proof}
    We first prove the lemma when $r$ satisfies $0 < c_1^r = c^2_r$. In this situation, $r$ first checks whether any \texttt{OFF}-colored robot exists in $Int(\Re)$. 
    If such a robot $r'$ exists, it must be in the move phase towards its target position. 
    Since $r$ lies on $\frac{1}{8}S_r$ in the current LCM cycle, it must have been qualified for the monitor robot when it was situated on $S_r$ and did not find \texttt{OFF}-colored robots in $Int(\Re)$. Thus, in the current LCM cycle of $r$, the robot $r'$ must be in the compute or move phase, aiming towards one of the longest sides of $\Re$ from its apex point $a_{r'}$. Therefore, $r$ eventually finds no \texttt{OFF}-colored robots in $Int(\Re)$. After this, according to our strategy, one of the shortest sides is selected as $SS_r$. There can be another robot $r''$ exists that satisfies all the conditions similar to those of $r$. Then, if it is activated simultaneously with $r$, it also chooses $SS_{r''}$. Let's consider $r_1$ as the neighbour of $r$ on the convex hull $\mathcal{CH}_r$. If the shortest side $AD$ lies within or intersects the half-plane defined by the line $\overleftrightarrow{rr_1}$, where no other robot of $\Re$ resides, then $r$ assigns $SS_r = AD$, and $r''$ assigns $SS_{r''} = BC$. Otherwise, $SS_r = BC$ and $SS_{r''} = AD$. Thus, the selection of $SS_r$ and $SS_{r''}$ are different. 
    Then, $r$ chooses its destination point on $\frac{1}{4}S_r$ at $\frac{len(S_r)}{2c^1_r}$ distance away from $SS_r$. 
    Similarly $r''$ chooses its destination point on $\frac{1}{4}S_r$ (as it is same as $\frac{1}{4}S_{r''}$) at $\frac{len(S_r)}{2c^1_r}$ distance away from $SS_{r''}$. 
    These two destination points are $len(S_r)-\frac{len(S_r)}{c^1_r}$ distance apart from each other, and the movement is free from any collision. 
    So $r$ and $r''$ both move to separate partitions, each with an area $\frac{len(SS_r)}{2} \times \frac{len(S_r)}{c^1_r} = \frac{len(AC) \cdot len(AB)}{2c_r^1}$.

    We now prove the lemma when $0 = c^2_r < c^1_r$. Similar to the previous arguments, $r$ must eventually see no \texttt{OFF}-colored robot on $Int(\Re)$. If another robot $r''$ satisfies conditions similar to those of $r$, they find $SS_r$ and $SS_{r''}$ individually. According to our algorithm, $r$ chooses its destination point on $\frac{1}{2}S_r$ at $\frac{len(S_r)}{2c^1_r}$ distance away from $SS_r$  and $r''$ chooses its destination point on $\frac{1}{2}S_r$ at $\frac{len(S_r)}{2c^1_r}$ distance away from $SS_{r''}$. By a similar argument, it can be proved that $r$ and $r''$ reach different partitions, each with the area $ len(SS_r)\times \frac{len(S_r)}{c^1_r}$.

    Now, we prove that the movement of $r$ is free from the collision. Without loss of generality, we assume that $r''$ lies on the right half plane delimited by $L_r$. Then $SS_r$ must lie on the left half plane delimited on $L_r$. Since $d(t_r, SS_r) < d(t_{r''}, SS_r)$ for $t_r, t_{r''}$ lying on $L = \frac{1}{2}S_r$ (or $\frac{1}{4}S_r$), $t_{r''}$ lies on the right of $t_r$ on $L$. If $t_{mid}$ and $r_{mid}$ are the midpoints of the line segments $\overline{t_r t_{r''}}$ and $\overline{rr''}$, respectively, the line $\overleftrightarrow{t_{mid} r_{mid}}$ separates the two paths $\overline{r t_r}$ and $\overline{r'' t_{r''}}$. Hence, the movement of $r$ is free from the collision even if $r''$ simultaneously moves with $r$.
\end{proof}

\begin{lemma}
    \label{off_to_terminate}
    A \texttt{OFF}-colored robot $r$, which lies on $S_r$ and satisfies one of the following conditions, terminates to a unique partition of $\Re$, and the movement is collision-free.\\
    (i) $r$ finds a \texttt{FINISH}-colored robot on $\frac{1}{4}S_r$ and some robots on $[S_r^{opp}, \frac{1}{4}S_r^{opp}]$.\\ (ii) $r$ finds a \texttt{FINISH}-colored robot on $\frac{1}{4}S_r^{opp}$, but not on $[S_r, \frac{1}{4}S_r]$.\\ (iii) $r$ finds a \texttt{FINISH}-colored robot on $\frac{1}{2}S_r$.
\end{lemma}
\begin{proof}
    If $r$ satisfies (i), it first identifies the side $SS_r$, as described in Case 2, and $r$  follows the Type II partitioning. Another robot $r'$ on $S_r$ might exist that satisfies conditions similar to those of $r$. 
    It also chooses $SS_{r'}$ if it is simultaneously activated with $r$. 
    If the point of intersection of $S_r$ and $AD$ is visible to $r$, it selects $AD$ as $SS_r$ and $r'$ selects $BC$ as $SS_{r'}$. Otherwise, $SS_r = BC$ and $SS_{r'} = AD$. Thus $SS_r$ and $SS_{r'}$ are different. $r$ sets $L = \frac{1}{4}S_r$ and computes the number $k$, $d$, its target point $t_r$ on $L$ such that $d(t_r, SS_r) = (2k+1)d$, as explained in Case 2. 
    Similarly, $r'$ can find its target point $t_{r'}$ on $L$ such that $d(t_{r'}, SS_{r'}) = (2k'+1)d$, where $k'$ is the length of the sequence of robots on $L$ considered by the robot $r'$ with respect to $SS_{r'}$. 
    Observe that $d(t_r, t_{r'}) = len(S_r) - 2(k+k'+1)d > 0$, as  $len(S_r) \geq 2(k+k')d+4d$ and $d \neq 0$. Hence, $r$ terminates at a unique partition of $\Re$.
    When $r$ satisfies (ii) and (iii), by similar arguments as above, we can prove that it terminates at a unique partition of $\Re$.

    Using analogous reasoning as presented in the proof of collision freeness in Lemma \ref{apex_to_termination}, we can demonstrate that the movement of $r$ is also free from collisions.
\end{proof}



\begin{theorem}
\label{theorem3}
    Our algorithm \textsc{Rectangle\_Partition} solves the uniform partitioning without collision for the rectangular region in $O(N)$ epochs, where $N$ is the total number of robots in $\Re$.
\end{theorem}

\begin{proof}
By Remark \ref{eventually_equal_or_empty}, eventually, either all the robots lie on one of the longest sides of $\Re$, or both of the longest sides have an equal number of robots. On both accounts, a robot $r$ must qualify as a monitor robot, which later moves to its apex point and finds either $0 < c_r^1 = c_r^2$ or $0 = c_r^2 < c_r^1$. 
By Lemma \ref{apex_to_termination}, it is ensured that $r$ terminates at a unique partition of $\Re$ (it lies either on $\frac{1}{2}S_r$ for Type I partitioning or on $\frac{1}{4}S_r$ for Type II partitioning).
After termination of $r$, if there is an \texttt{OFF}-colored robot $r'$ on one of the longest sides, the position of $r$ helps it find its destination point either on $\frac{1}{2}S_r$ or on $\frac{1}{4}S_r$. Lemma \ref{off_to_terminate} ensures that $r'$ terminates at a unique partition of $\Re$.

    In the worst case, initially, all robots lie on a line inside $\Re$. It takes $O(N)$ epochs for all the robots to reach the boundary, as the movement of the robots will be sequential. Monitor robots on the sides of $\Re$ move to their apex point in one epoch. Moreover, if the two longest sides of $\Re$ have different numbers of robots, the robots move from the longest side with fewer robots (say $AB$) to the other longest side ($CD$) of $\Re$. A monitor robot $r$ on $AB$ gets activated and moves to its apex points in one epoch. 
    Since the number of robots on $CD$ is more, $r$ needs to wait till all the robots on $[S_r^{opp}, \frac{1}{8}{S_r^{opp}}]$ move back to $CD$, which requires four epochs, as each of the two \texttt{FINISH}-colored robots lying on $\frac{1}{8}CD$ might take one epoch to reach $CD$ and one epoch to change its color to $\texttt{OFF}$. 
    It might also need to wait for one epoch if there is another robot on $(\frac{1}{8}S_r, \frac{7}{8}S_r)$ moving towards $CD$. When no such robot exists on $(\frac{1}{8}S_r, \frac{7}{8}S_r)$, $r$ calculates $c_r^1$ and $c_r^2$ (by Remark \ref{remark_1_rectangle}) and decides to move to $CD$ which takes one epoch. So, it needs seven epochs for $r$ to move to $CD$ from $AB$. In the worst case, it is possible that the robots on $AB$ sequentially execute the movement to $CD$. So, this process takes $O(N)$ epochs to complete. So, all the robots reposition themselves on exactly one longest side or get distributed equally among the two longest sides of $\Re$ takes $O(N)$ epochs. When all the robots are on one longest side $CD$ of $\Re$, there can be two monitor robots $r$ and $r'$. They reach their respective apex points in one epoch. In the worst case, $r'$ gets activated first and moves to its final position on $\frac{1}{2}CD$. While $r'$ is executing its movement, $r$ gets activated and finds $r'$ on $(\frac{1}{8}CD, \frac{7}{8}CD)$ (as $S_r = CD$). So, it decides to move back to $CD$ with color \texttt{FINISH}, which needs one epoch. When it reaches $CD$, it needs another epoch to change its color to \texttt{OFF}. When $r$ is again activated and finds $r'$ on $\frac{1}{2}S_r$, it calculates and moves to its final position on $\frac{1}{2}S_r$ in one epoch.  Thus, $r$ needs at most four epochs to reach its final point. In the worst case, the other \texttt{OFF}-colored robot on $CD$ moves to their final point one by one, requiring $O(N)$ epochs. 
    On the other hand, if the robots are distributed equally on both the longest sides, by similar arguments as above, we can prove that the robots need $O(N)$ epochs to terminate. 
    So, overall, the algorithm requires $O(N)$ epochs to achieve uniform partitioning when the region is rectangular.

    Lemma \ref{lemma3.1}, \ref{newlemma1_rectangle}, \ref{lemma3.3}, \ref{apex_to_termination} and  \ref{off_to_terminate} ensure that the movements of the robots are collision-free.
\end{proof}

\section{Algorithm for a Square Region}
\label{square}

The region $\Re$ is considered to be square in this section. Our proposed algorithm \textsc{Square\_Partition} uses $5$ colors described in the Table \ref{color_table} with their specification. The pseudocode of the algorithm is included in the \ref{appendix}.

\subsection{Description of the Algorithm}
 From any initial deployment of the robots inside $\Re$, our target is to move the interior and corner robots to the boundary.  
A robot $r$, after activating with color \texttt{OFF}, determines whether it is a corner, boundary or interior robot. In case of $r$ being a corner or an interior robot, it finds a target side $S_r$ and moves to a point on it with maintaining its color \texttt{OFF}. Then $r$ waits for other interior robots and the visible corner robots to move to one of the sides of $\Re$. 

 If $r$ is a corner robot, it chooses any of the incident sides as $S_r$. 
 In the case of $r$ being an interior robot on $\Re$, it selects one of the nearest sides from its current position as $S_r$. 
 The strategy is the same as the rectangle. Since there is no shortest side in a square, the difference here is when $r$ is a boundary robot in $\Re$, it does not move. 
 In all other cases, a target point $t_r$ is calculated on $S_r$, and the robot $r$ moves to it with color \texttt{OFF} by following a movement strategy similar to the strategy \textsc{Move\_To\_LongestSides}.

 After this, the outline of the strategy for the square region has similarities with the algorithm \textsc{Rectangle\_Partition}. In rectangular regions, we ask the monitor robots to move at a particular distance from the longest sides and make them compare the number of robots near the two longest sides of the region. If the number is not the same, we bring the robots from the side with fewer robots to the side with a larger number of robots. In the case of the square region, we follow a similar strategy, but here, the robots can be situated on any side of $\Re$. If all the robots are distributed among the two sides $S_r$ and $S_r^{opp}$, the robot $r$ follows the Type I or the Type II partitioning. For other cases, $r$ follows Type III and Type IV partitioning.  
 We define some new notations and conventions to describe the algorithm for the square regions. 
 \begin{itemize}
     \item $C$ is the center of the square $\Re$. 
     
     \item The triangle $\Delta Ce^1_Se^2_{S}$, denoted by $\Delta S$, is called the \emph{side triangle} of the side $S$ where $e_{S}^1$ and $e_S^2$ are the endpoints of the side $S$. 
     
     \item When we say that $r$ lies on  $\Delta S$, we mean that $r$ lies  either on the side $S$ or strictly within the triangle $\Delta Ce^1_Se^2_{S}$, but not on $\overline{Ce^1_S}$ or $\overline{Ce^2_S}$.
     
     \item $|S|$ denotes the number of robots lying on the side triangle $\Delta S$ of $S$.

 
     \item For the robot $r$, $S_r^L$ and $S_r^R$ represent the sides of $\Re$ other than $S_r$ and $S_r^{opp}$.

     \item $S_r$ is the nearest side of the square $\Re$ for the robot $r$.
     
     \item $D_r$ and $D_r^{opp}$ are the two diagonals of $\Re$ for the robot $r$.
     
     \item $Max_r$ represents a set of all sides with the maximum number of robots lying on their corresponding side triangle.
 \end{itemize}

\noindent $Max_r$ helps $r$ to decide which side triangles contain the maximum number of robots. For example, 
\begin{wrapfigure}[11]{r}{0.4\textwidth}
\centering
  \includegraphics[width=0.9\linewidth]{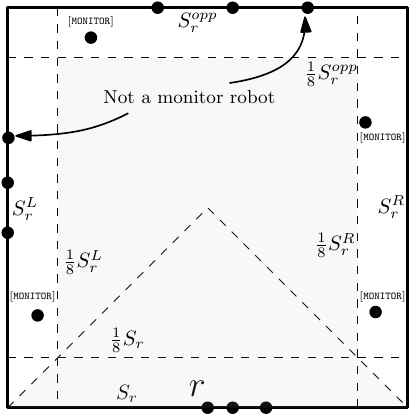}
     \caption{$r$ qualifies to be a monitor robot}
     \label{monitor_robot_square}
\end{wrapfigure}
$Max_r = \{S_r, S_r^{opp}\}$ if $|S_r| = |S_r^{opp}|  > $ $|S_r^L|,~ |S_r^R|$. Here, the sides $S_r$ and $S_r^{opp}$ have the maximum number of robots. Our aim is to move all the robots from the sides not in $Max_r$ to the sides in $Max_r$. When $Max_r$ contains all the sides of $\Re$, all the four sides have $\frac{N}{4}$ robots.

 Now, we redefine the monitor robot for the square region. 
 A robot $r$ on $S_r$ is called a \emph{monitor robot} (refer to Fig. \ref{monitor_robot_square}) if all the four conditions are satisfied. (i) $r$ is a terminal robot on $S_r$. (ii) $ \left((S_r, \frac{1}{8}S_r) \cap \Delta S_r \right) \cup \left( (\frac{1}{8}S_r, \frac{7}{8}S_r) \cap (\frac{1}{8}S_r^L, \frac{7}{8}S_r^L)\right)$ has no robots (shown as the shaded region). (iii) All the robots  on $(S_r^{opp}, \frac{1}{8}S_r^{opp}] \cup (S_r^L, \frac{1}{8} S_r^L] \cup (S_r^R, \frac{1}{8}S_r^R]$ are with color \texttt{MONITOR}. (iv) There is no visible corner robot.  

 We differentiate the rest of the algorithm into three major cases for a terminal robot $r$ lying on the side $S_r$.

\noindent \textbf{Case 1 ($r$ is a monitor robot and $Int(\Re)$ has no robot with color \texttt{FINISH} or \texttt{FINISH1} or \texttt{FINISH2}):}
Computing $Max_r$ is possible when $r$ moves to a point in $Int(\Re)$, referred to as \emph{apex point}, to calculate the number of robots on each side. 
The apex point must be chosen so that $r$ gets to see all the robots in $\Re$, and it does not obstruct the visibility of any other robot after the movement. 
Also, we want $r$ to lie within $\Delta S_r$ so that it can choose $S_r$ easily. If $S_r$ has no other robot, it does not need to move to an apex point. 
Rather, it can find out $Max_r$ without moving from $S_r$ and directly follow the movement strategy for the \texttt{MONITOR}-colored robots lying on $(S_r, \frac{1}{8}S_r]$, which is described later in this case. We now discuss the movement strategy of a monitor robot $r$ lying on $S_r$ to its apex points where $|S_r|\geq 2$.

\noindent \textbf{Strategy} \textsc{MonitorRobot\_Movement\_Square}: If the robot $r$ does not find any robot on 
\begin{wrapfigure}[12]{r}{0.4\textwidth}
\centering
  \includegraphics[width=0.9\linewidth]{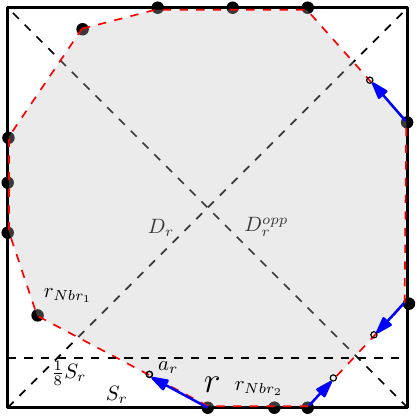}
     \caption{$r$ moves to its apex point $a_r$ considering the convex hull $\mathcal{CH}_r$}
     \label{monitor_robot_movement_square}
\end{wrapfigure}
$\Delta S_r^{opp}$ nor $\Delta S_r^{L}$ nor $\Delta S_r^R$, it chooses the point of intersection of $L_r$ and $\frac{1}{8}S_r$ as its apex point $a_r$ and moves to it with color \texttt{MONITOR}. Otherwise, $r$ finds the neighbour $r_{Nbr_1}$ on the convex hull $\mathcal{CH}_r$, which does not lie on the side $S_r$. $r$ also finds the neighbour $r_{Nbr_2}$ on $\mathcal{CH}_r$ other than $r_{Nbr_1}$ which lies on $S_r$, as shown in Fig \ref{monitor_robot_movement_square}. It calculates the apex point $a_r$ on the line segment $\overline{rr_{Nbr_1}}$ such that $d(r,a_r) = \frac{1}{2} \min \{d(r,D_r), d(r, D_r^{opp}), d(r,\frac{1}{4}S_r), d(r, r_{Nbr_2}) \}$. The point $a_r$ is chosen on $\overline{rr_{Nbr_1}}$ to ensure that $r$ does not become an obstruction for other monitor robots after its movement. Moreover, it also ensures that $r$ does not cross $D_r$ and $D_r^{opp}$, allowing $r$ to accurately determine $S_r$ (so that it remains the nearest side to $r$) after moving to the apex point.
We also want $r$ to lie on or below the line segment $\frac{1}{8}S_r$, as our aim is to terminate $r$ either on $\frac{1}{2}S_r$, $\frac{1}{4}S_r$, $\frac{1}{3}S_r$ or $\frac{1}{6}S_r$. 
Additionally, we do not cross the line $L_{r_{Nbr_2}}$ after the movement to ensure that all other robots on $S_r$ must be on one side of the half-plane delimited by $L_r$. Finally, $r$ changes its current color to \texttt{MONITOR} and moves to $a_r$.

Next, we explain how a \texttt{MONITOR}-colored robot $r$ lying on $(S_r, \frac{1}{8}S_r]$, decides its destination based on $Max_r$. It moves from its apex point either to one of the sides or to the final position. 

\noindent \textbf{Movement of a \texttt{MONITOR}-colored Robot $r$ lying on $(S_r, \frac{1}{8}S_r]$:} If $r$ sees any \texttt{OFF}-colored robot in $Int(\Re)$, $r$ does not change its color or position. Otherwise, before calculating $Max_r$, $r$ checks whether there is any visible \texttt{FINISH} or \texttt{FINISH1} or \texttt{FINISH2}-colored robot. These three colors are used when a robot moves to its final position.   
Such a robot might lead to the miscalculation of $Max_r$, as it can become an obstruction for $r$. 
We explain the rest of the algorithm using a proper example and figures for better comprehension. We consider a square with corners $e_P$ (common corner of $S_r$ and $S_r^L$), $e_Q$ (common corner of $S_r$ and $S_r^R$), $e_R$ (common corner of $S_r^{opp}$ and $S_r^R$) and $e_X$ (common corner of $S_r^{opp}$ and $S_r^L$), as depicted in Fig. \ref{new_case1.2_sqaure}. 
    We now identify four sub-cases.
\begin{itemize}
    \item \textbf{Case 1.1 ($r$ sees a \texttt{FINISH}-colored robot):} $r$ waits till all \texttt{FINISH}-colored robots reach $\frac{1}{2}S_r$ or $\frac{1}{4}S_r$ or $\frac{3}{4}S_r$. 
    The robot $r$ now moves to its final position on $L = \frac{1}{2}S_r$ or $\frac{1}{4}S_r$ with color \texttt{FINISH} by following the same movement strategy as an \texttt{OFF}-colored robot, described in Case 2 of \textsc{Rectangle\_Partition} algorithm (Section \ref{rectangle}). The only difference is the choice of $SS_r$, which, in this case, belongs to $\{ S_r^L, S_r^R \}$ (as there is no shortest side in a square).

    \item \textbf{Case 1.2 ($r$ sees a \texttt{FINISH1}-colored robot):} $r$ waits till all \texttt{FINISH1}-colored robots reach $\frac{1}{3}S_r$ or $\frac{1}{3}S$ where $S \in \{S_r^L, S_r^R \}$. 
    Let $c$ be number of \texttt{FINISH1}-colored robot on $\frac{1}{3}S_r$. Note that $c \leq 1$ in this case. Let us also assume that $c'$ is the number of robots lying on $\Delta S_r$ without the color \texttt{FINISH1}. 
    Without loss of generality, let us assume $S = S_r^L$. $r$ first chooses the corner $e_R$ as the vertex of Type III partitioning and then calculates two points $A$ and $B$ on $S_r$ which are $\frac{len(S_r)}{c+ c'}$ distance away from $e_P$ and $e_Q$, respectively, as shown in Fig. \ref{new_case1.2_sqaure}. 
    If $S^{opp}$ lies on the half plane delimited by $L_r$ where other robots on $\Delta S_r$ lie, $r$ chooses the triangle $\mathcal{T} = \Delta Ae_Pe_R$. Otherwise, it chooses $\mathcal{T} = \Delta Be_Qe_R$. $r$ moves to the centroid of the triangle $\mathcal{T}$ after changing its current color to \texttt{FINISH1} from \texttt{FINISH}.
        \begin{figure}[h]
 \begin{minipage}[c]{0.47\textwidth}
     \centering
     \includegraphics[width=0.8\linewidth]{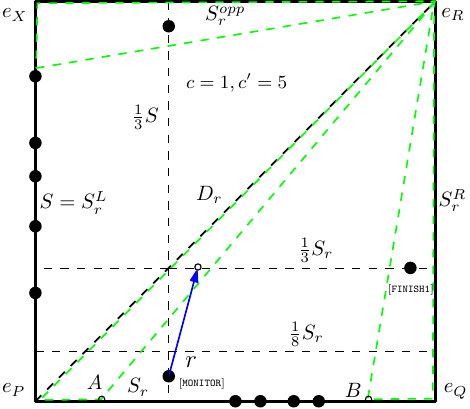}
     \caption{The movement of $r$ for Type III partitioning when $r.color = $ \texttt{MONITOR}}
     \label{new_case1.2_sqaure}
 \end{minipage}
 \hfill
 \begin{minipage}[c]{0.47\textwidth}
     \centering
     \includegraphics[width=0.8\linewidth]{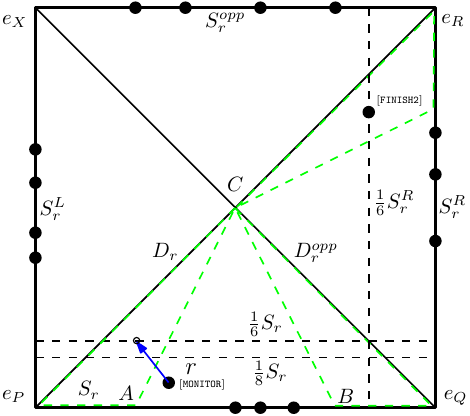}
     \caption{The movement of $r$ for Type IV partitioning when $r.color = $ \texttt{MONITOR}}
     \label{new_case1.3_square}
 \end{minipage}
 \end{figure}

    \item \textbf{Case 1.3 ($r$ sees a \texttt{FINISH2}-colored robot):} $r$ waits till all the \texttt{FINISH2}-colored robots lie on  $\frac{1}{6}S$ where $S \in \{S_r, S_r^{opp}, S_r^L, S_r^R\}$. Afterwards, $r$ computes the two points $A$ and $B$ on $S_r$ which are $\frac{len(S_r)}{|S_r|}$ distance away from $e_P$ and $e_Q$, respectively as shown in Fig. \ref{new_case1.3_square}.
     If $e_P$  and all the other robots on $\Delta S_r$ lie on the different half plane delimited by $L_r$, the triangle $\mathcal{T}$ is chosen that satisfies $\mathcal{T} = \Delta Ae_PC$.
    Otherwise, it chooses $\mathcal{T} = \Delta Be_QC$. Then $r$ moves to the centroid of the triangle $\mathcal{T}$ after changing its current color to \texttt{FINISH2} from \texttt{FINISH}.

    \item \textbf{Case 1.4 ($r$ does not see any \texttt{FINISH} or \texttt{FINISH1} or \texttt{FINISH2}-colored robot):} In this case,  $r$ considers $Max_r$ to understand the type of partitioning of the region $\Re$.
    We identify ten sub-cases, out of which the first four sub-cases discuss the process of choosing the type of partitioning and the movement of $r$ to the final position in its respective partition.

\begin{itemize}
    \item \textbf{Case 1.4.1 ($Max_r = \{S_r \}$ and $\Delta S$ contains no robot for all $S\in \{ S_r^{opp}, S_r^L, S_r^R\}$):} Similar to Case 1.4 in the algorithm \textsc{Rectangle\_Partition} (in Section \ref{rectangle}), $r$ moves to the final position on $\frac{1}{2}S_r$ with color \texttt{FINISH} by considering $SS_r $ from  $ \{S_r^L, S_r^R \}$.

    \item \textbf{Case 1.4.2 ($Max_r = \{S_r, S_r^{opp}\}$ and no robots on $\Delta S$ for all $S \in \{S_r^L, S_r^R \}$):} Similar to Case 1.3 in the algorithm \textsc{Rectangle\_Partition} (in Section \ref{rectangle}), $r$ moves to the final position on $\frac{1}{4}S_r$ with color \texttt{FINISH} by considering $SS_r $ from  $ \{S_r^L, S_r^R \}$.

    \item \textbf{Case 1.4.3 ($Max_r = \{ S_r, S \}$ and both $\Delta S_r^{opp}$ and $\Delta S^{opp}$ contain no robot where $S \in \{S^L_r, S^R_r\}$):}  Without loss of generality, let us assume that $S = S_r^L$. Similar to Case 1.2 of \textsc{Square\_Partition} algorithm, $r$ calculates two points $A$ and $B$ on $S_r$, as illustrated in Fig \ref{new_case1.2_sqaure}. 
    It selects the triangle $\mathcal{T}$ and moves to the centroid of $\mathcal{T}$ after changing its color to \texttt{FINISH1} from \texttt{MONITOR}. 

    \item \textbf{Case 1.4.4 ($Max_r = \{ S_r, S_r^{opp}, S_r^L, S^R_r \}$):} 
    Similar to Case 1.3 of \textsc{Square\_Partition} algorithm, $r$ finds the triangle $\mathcal{T}$ and moves to the centroid of the triangle $\mathcal{T}$, as depicted in Fig \ref{new_case1.3_square} after changing its current color to \texttt{FINISH2} from \texttt{MONITOR}.
\end{itemize}


\noindent For the remaining six sub-cases, $r$ chooses a target side before its movement. Our aim is to gather all the robots to one of the four sides of $\Re$. If, due to symmetry, it is not possible, robots should be gathered on two sides of $\Re$.
When Case 1.4.1, Case 1.4.2, Case 1.4.3 and Case 1.4.4 do not hold, the following cases may occur for different $Max_r$. The robot $r$ decides a side as its target for its movement based on $Max_r$.
\begin{itemize}
    \item \textbf{Case 1.4.5 ($Max_r = \{ S_r, S_r^{opp}, S\}$ where $S \in \{S_r^L , S^R_r\}$):} Without loss of generality, if $S = S_r^L$, the side $S_r^R$ has the minimum number of robots. We target to move $r$ to the side $S_r^L$, which is opposite to the side having minimum number of robots. So, $r$ sets the side $S$ as its target side. 

    \item \textbf{Case 1.4.6 ($Max_r = \{ S \}$ or $\{S, S'\}$, where $S, S' \in \{S_r^L , S^R_r\}$ with $S \neq S'$):} In this case, either one or two adjacent sides to $S_r$ contain maximum number of robots. So $r$ targets to move on one of them accordingly. The target side is chosen to be the side $S$.

     \item \textbf{Case 1.4.7 ($Max_r =  \{ S_r^L, S^R_r, S_r^{opp}\}$ or \{$S_r^{opp} \})$:} The target side is  the side $S_r^{opp}$.

     \item  \textbf{Case 1.4.8 ($Max_r = \{ S, S_r^{opp}\}$, where $S \in \{S_r^L , S^R_r\})$):} The target side for $r$ is $S$.

 \begin{figure}[h]
 \begin{minipage}[c]{0.47\textwidth}
     \centering
  \includegraphics[width=0.7\linewidth]{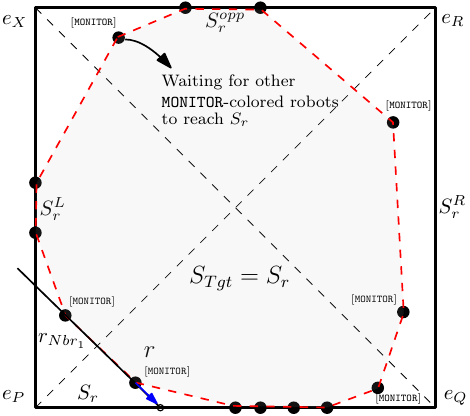}
     \caption{$r$ moves back to $S_r$ from its apex point $a_r$ considering the convex hull $\mathcal{CH}_r$}
     \label{move_to_targetside_sqaure}
 \end{minipage}
 \hfill
 \begin{minipage}[c]{0.47\textwidth}
     \centering
     \includegraphics[width=0.7\linewidth]{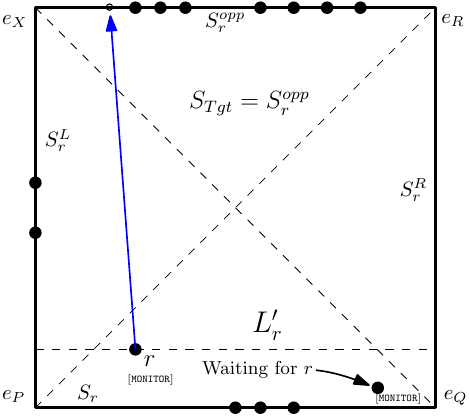}
     \caption{The movement of $r$ to $S_{Tgt}$ when the half plane delimited by $L'_r$ has no \texttt{MONITOR}-colored robot}
     \label{move_to_targetside_square1}
 \end{minipage}
 \end{figure}
     
    \item \textbf{Case 1.4.9 ($Max_r = \{ S_r \}$ or $\{ S_r, S_r^L, S_r^R \}$):} In this case, $r$ chooses $S_r$ its target side.

    \item \textbf{Case 1.4.10 ($Max_r = \{S, S_r\}$ where $S \in \{S_r^{opp}, S_r^L , S^R_r\}$):} $S_r$ is the target side of $r$.
\end{itemize}
\end{itemize}

Let $S_{Tgt}$ be the target side of $r$ and $L_r'$ be the line passing through $r$ and parallel to $S_{Tgt}$. 
If $S_{Tgt} = S_r$, then $r$ considers the line $\overleftrightarrow{rr_{Nbr_1}}$ (ref. Fig \ref{move_to_targetside_sqaure}) where $r_{Nbr_1}$ is the neighbour of $r$ on the convex hull $\mathcal{CH}_r$ not lying on $\Delta S_r$. It moves to the point of intersection $t_r$ of $S_r$ and $\overleftrightarrow{rr_{Nbr_1}}$ with the current color \texttt{MONITOR}. If $r$ gets activated again with color \texttt{MONITOR} on $S_r$, it changes its color to \texttt{OFF} with no movement.
When $S_{Tgt}\neq S_r$, $r$ waits until there is no \texttt{MONITOR}-colored robot lying on the half plane delimited by $L_r'$ that contains $S_{Tgt}$, as shown in Fig. \ref{move_to_targetside_square1}. $r$ moves with color \texttt{OFF} to the target side $S_{Tgt}$, by following the similar movement strategy as \textsc{Movement\_To\_LongestSide}. The above cases can interchangeably occur in different LCM cycles of a robot $r$.

\noindent \textbf{Case 2 ($Int(\Re)$ has \texttt{FINISH}-colored robots):} $r$ understands that the partitioning of $\Re$ would be of Type II after seeing a \texttt{FINISH}-colored robot. $r$ figures out its final position by looking at the positions of the robots with color \texttt{FINISH}. $r$ follows the movement strategy similar to Case 2 in \textsc{Rectangle\_Partition} algorithm (Section \ref{rectangle}) to move to the final position on $\frac{1}{2}S_r$ (for Type I) or on $\frac{1}{4}S_r$ (for Type II) with color \texttt{FINISH}.
 
\noindent \textbf{Case 3 ($Int(\Re)$ has \texttt{FINISH1}-colored robots):} $r$ understands the partitioning would be of Type III after seeing \texttt{FINISH1}-colored robots. $r$ figures out its final position by looking at the positions of the robots with color \texttt{FINISH1}.

\begin{itemize}
    \item \textbf{$\frac{1}{3}S_r^L$ has \texttt{FINISH1}-colored robot:} $r$ chooses the side $S=S_r^L$.
    
    \item  \textbf{$\frac{1}{3}S_r^R$ has a \texttt{FINISH1}-colored robot:} $r$ chooses $S_r^R$ as $S$.

    \item \textbf{A robot lying on the side triangle of $\Delta S_r^L$:} $r$ chooses the side $S_r^L$ as $S$.

    \item \textbf{A robot lying on the side triangle of $\Delta S_r^R$:} The side $S_r^R$ is chosen as $S$.
\end{itemize}

The robot $r$ waits till all the \texttt{FINISH1}-colored robots lie either on $\frac{1}{3}S_r$ or on $\frac{1}{3}S$. Without loss of generality, we assume $S = S_r^L$. At this point, $r$ understands that the partitioning is of Type III. It chooses the point of intersection of $S_r^{opp}$ and $S^{opp}$ (i.e., the corner $e_R$ in this scenario) as the vertex of Type III partitioning. 
$D_r$ is the diagonal of the region $\Re$ passing through the vertex of Type III partitioning and $D_r^{opp}$ is the other diagonal of $\Re$.
Now it computes the length of the base $bl$ of each triangular partition depending on the positions of the \texttt{FINISH1}-colored robots, as shown in Fig. \ref{fig.case2square}. So, there could be two sub-cases.

\begin{itemize}
    \item \textbf{Case 3.1 (There is a \texttt{FINISH1}-colored robot on $\frac{1}{3}S_r$):} Let $r'$ be a terminal robot on $\frac{1}{3}S_r$ with color \texttt{FINISH1}. Let $bl_1$ be the length of the base of the triangle whose one side is $D_r$ and the centroid is $r'$. $bl_2$ is the length of the base of the triangle whose one side is $S^{opp}$ and the centroid is at $r'$. So, $r$ computes $bl$ such that  $bl = \min \{bl_1 , bl_2\}$.

    \item \textbf{Case 3.2 (There is no \texttt{FINISH1}-colored robot on $\frac{1}{3}S_r$):} In this case \texttt{FINISH1}-colored robot (say $r'$) must be on $\frac{1}{3}S$. Like the previous case, $bl$ is computed by finding $bl_1$ and $bl_2$. The definition of $bl_1$ is the same as in the previous case (Case 3.1). We calculate $bl_2$ by considering the triangle with one side of $S_r^{opp}$  and the centroid at $r'$.
\end{itemize}

 
Let $u_1$ be the point of intersection of $\frac{1}{3}S_r$ and the diagonal $D_r$. $u_2$ is the point of intersection 
\begin{wrapfigure}[12]{r}{0.4\textwidth}
     \centering
     \includegraphics[width=\linewidth]{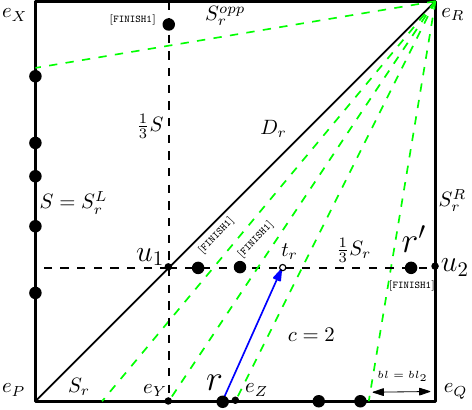}
     \caption{$r$'s movement for Type III partitioning for $r.color = $ \texttt{OFF}}
     \label{fig.case2square}
\end{wrapfigure}
of $\frac{1}{3}S_r$ and $S^{opp}$.
If the side $S^{opp}$ lies in the half plane delimited by $\overleftrightarrow{re_R}$, where the other robots on $S_r$ reside, $r$ finds the number of \texttt{FINISH1}-colored robots $c$ on $\frac{1}{3}S_r$ starting from the robot ${bl}/3$ distance apart from $u_1$ towards $u_2$ such that two consecutive robots are at ${2bl}/3$ distance away from each other, as depicted in Fig. \ref{fig.case2square}. Then, it changes its color to \texttt{FINISH1} and moves to the point $t_r$ on $\frac{1}{3}S_r$ such that $t_r = Centroid(\Delta e_Y e_R e_Z)$, where $e_Y$ and $e_Z$ are the points on $S_r$ satisfying $d(e_Y, e_P) = c\cdot bl$ and $d(e_Z, e_P) = (c+1) \cdot bl$. Otherwise, $c$ is calculated as the number of robots on $\frac{1}{3}S_r$ starting from the robot ${bl}/3$ distance apart from $u_2$ towards $u_1$ to a robot such that two consecutive robots are at ${2bl}/3$ distance away from each other. $r$ moves to $t_r$ on $\frac{1}{3}S_r$ such that $t_r = Centroid(\Delta e_Y e_R e_Z)$, where $e_Y$ and $e_Z$ are the points on $S_r$ satisfying $d(e_Y, e_Q) = c\cdot bl$ and $d(e_Z, e_Q) = (c+1) \cdot bl$ with color \texttt{FINISH1}. 

\noindent \textbf{Case 4 ($Int(\Re)$ has \texttt{FINISH2}-colored robots):}  In this case, $r$ understands that the partitioning would be of Type IV after seeing \texttt{FINISH2}-colored robots. $r$ waits till any robot with color \texttt{FINISH2} lying on the side triangle of a side $S$ reaches $\frac{1}{6}S$, where $S \in \{ S_r, S_r^{opp}, S_r^L, S_r^R\}$. It now needs to calculate the base length $bl$ of the triangular partition.
Let $r'$ be a \texttt{FINISH2}-colored terminal robot on $\frac{1}{6}S_{r'}$, in Fig. \ref{fig.case3square}.
Let us also consider that ${bl_1}$ is the length of the base of the triangle
\begin{wrapfigure}[11]{r}{0.4\textwidth}
\centering
     \includegraphics[width=\linewidth]{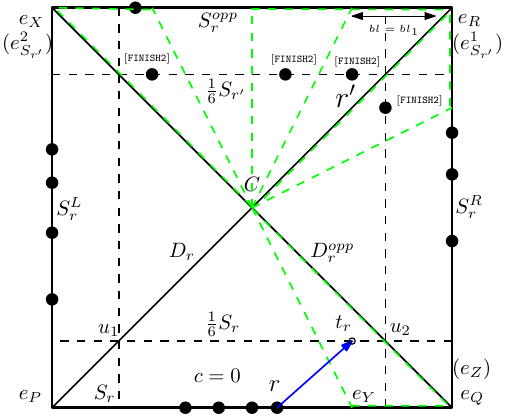}
     \caption{$r$'s movement for Type IV partitioning for $r.color = $ \texttt{OFF}}
     \label{fig.case3square}
    \end{wrapfigure}
whose centroid is $r'$, and the two vertices are $C$ and $e_{S_{r'}}^1$. ${bl_2}$ is the length of the base of the triangle whose centroid is $r'$ and the two vertices are $C$ and $e_{S_{r'}}^2$. Now, $r$ calculates the base length $bl = \min \{ bl_1, bl_2\}$. 
$D_r$ and $D_r^{opp}$ are the two diagonals of passing through $e_P$ and $e_Q$ respectively.
Let us further assume $u_1$ (and $u_2$) is the point of intersection of $\frac{1}{6}S_r$ and $D_r$ (and $D_r^{opp}$). If the line segment $\overline{Ce_Q}$ lies on the half plane delimited by the line $\overleftrightarrow{rC}$, where other robots on $S_r$ reside, $r$ finds the number of \texttt{FINISH2}-colored robots $c$ on $\frac{1}{6}S_r$ starting from the robot ${bl}/3$ distance away from $u_1$ towards $u_2$ such that two consecutive robots are ${2bl}/3$ distance apart from each other. Finally, $r$ changes its color to \texttt{FINISH2} and moves to the point $t_r$ such that $t_r = Centroid(\Delta C e_Y e_Z)$, where $e_Y$ and $e_Z$ are the points on $S_r$ satisfying $d(e_P, e_Y) = c\cdot bl$ and $d(e_P, e_Z) = (c+1) \cdot bl$. Otherwise, if the line segment $\overline{Ce_Q}$ lies on the other half plane delimited by the line $\overleftrightarrow{rC}$, where other robots on $S_r$ reside, $r$ finds the number of \texttt{FINISH2}-colored robots $c$ on $\frac{1}{6}S_r$ starting from the robot ${bl}/3$ distance away from $u_2$ towards $u_1$ such that two consecutive robots are ${2bl}/3$ distance apart from each other. Finally, $r$ changes its color to \texttt{FINISH2} and moves to the point $t_r$ such that $t_r = Centroid(\Delta C e_Z e_Y)$, where $e_Z$ and $e_Y$ are the points on $S_r$ satisfying $d(e_Q, e_Z) = c\cdot bl$ and $d(e_Q, e_Y) = (c+1) \cdot bl$.

\subsection{Analysis of the Algorithm}
In this subsection, we analyse the algorithm \textsc{Square\_Partition}. The movement of the robots is free from collision. The following lemmas and theorems provide the correctness and time complexity of the algorithm. 
\begin{lemma}
    \label{lemma4.1}
    All interior and corner robots move to the boundary without collision. 
\end{lemma}
\begin{proof}
    The proof of this lemma follows from Lemma \ref{lemma3.1}.
\end{proof}

\begin{lemma}\label{move_to apex point_square}
    A monitor robot $r$ lying on $S_r$ selects its apex point $a_r$ on $(S_r, \frac{1}{8}S_r] \cap \Delta S_r$  so that all the robots on $\Re$ are visible to it after its movement and moves to it without collision.
\end{lemma}

\begin{proof}
    When $r$ qualifies as a monitor robot on $S_r$, it finds the apex point $a_r$ such that $d(r, a_r) < d(r, D_r)$ and $d(r,a_r) < d(r, D_r^{opp})$. So, $a_r$ cannot lie on $D_r$ or $D_r{opp}$ or outside $\Delta S_r$. Also, $d(r,a_r) \leq d(r, \frac{1}{8}S_r)$. Thus $a_r$ must be on  $(S_r, \frac{1}{8}S_r]$. Hence $a_r$ lies on $(S_r, \frac{1}{8}S_r] \cap \Delta S_r$. The point $a_r$ is chosen on the line segment $rr_{Nbr_1}$ where $r_{Nbr_1}$ is the neighbour of $r$ on $\mathcal{CH}_r$ which does not lie on $S_r$. By the similar argument as presented in Lemma \ref{lemma 3.2}, we can conclude that $r$ remains a vertex of $\mathcal{CH}^G$ (which is the convex hull of all the robots in $\Re$) after its movement which enables it to see all the robots in $\Re$ after reaching to its apex point. 

    Observe that even if $r$ and $r_{Nbr_1}$ move simultaneously on $\overline{rr_{Nbr_1}}$, they do not collide as their respective apex points are separated by either $D_r$ or $D_r^{opp}$.  Additionally, another monitor robot lying on a different side triangle other than $\Delta S_r$ cannot collide with $r$ as its movement to its apex point is restricted within its respective side triangle. If $r'$ is another monitor robot on $S_r$ moving simultaneously with $r$, they move to their respective apex points by following two different sides of the convex hull $\mathcal{CH}^G$ similar to Lemma \ref{newlemma1_rectangle}, resulting to no collision.    
\end{proof}

\begin{remark}\label{max_r_accurate_calculation}
    When $r$ finds itself on $(S_r, \frac{1}{8}S_r]$ with color \texttt{MONITOR} and does not see any \texttt{OFF} or \texttt{FINISH} or \texttt{FINISH1} or \texttt{FINISH2} robot in $Int(\Re)$, it accurately calculates $Max_r$.
\end{remark}

\begin{remark}\label{one_or_two_sideconfiguration}
    Let $r$ be the first robot to set its color to \texttt{FINISH}. Then, all the robots are positioned on one side of $\Re$ or the two opposite sides of $\Re$ before the movement of $r$ with color \texttt{FINISH}, according to our strategy described in Case 1.4.1 and 1.4.2.
\end{remark}

\begin{remark}\label{adjacent_sideconfiguration}
    Let $r$ be the first robot to set its color to \texttt{FINISH1}. Then, all the robots are positioned on the two adjacent sides of $\Re$ before the movement of $r$ with color \texttt{FINISH1}, according to our strategy described in Case 1.4.3.
\end{remark}

\begin{remark}\label{four_sideconfiguration}
    Let $r$ be the first robot to set its color to \texttt{FINISH2}. Then, all the robots are distributed equally on the four sides of $\Re$ before the movement of $r$ with color \texttt{FINISH2}, according to our strategy described in Case 1.4.4.
\end{remark}

\begin{lemma}
    \label{lemma4.2}
    All robots form either a one-side configuration (where all robots gather on one side), or a two-side configuration (where robots gather on either two adjacent or two opposite sides), or a four-side configuration (where robots are equally distributed on four sides).
\end{lemma}

\begin{proof}
    After reaching all the robots on the boundary, a robot $r$ will move to the apex point to calculate $Max_r$. $r$ first checks whether there is any \texttt{OFF}-colored robot in $Int(\Re)$. 
    If yes, it means that the \texttt{OFF}-colored robot is in its move phase towards one of the sides from its apex point. 
    So $r$ eventually finds no robot with color \texttt{OFF} in $Int(\Re)$. 
    We assume that there is neither a \texttt{FINISH}, nor \texttt{FINISH1} nor a \texttt{FINISH2} robot in $Int(\Re)$. If such a robot exists, then the lemma holds from Remarks \ref{one_or_two_sideconfiguration}, \ref{adjacent_sideconfiguration} and \ref{four_sideconfiguration}. 
    If $Max_r$ has only one element (side), $r$ moves towards the side with maximum robots (follows from Case 1.4.6 and Case 1.4.9), which leads to gathering all robots on one side. 
    If $Max_r$ contains two sides, then $r$ moves to one of those two sides in $Max_r$ (follows from Case 1.4.6, 1.4.8, and 1.4.10), leading all the robots lying on two sides of $\Re$. If $Max_r$ has three sides, our target is to move $r$ to the side opposite to the side having minimum robots (follows from Case 1.4.5, 1.4.7, and 1.4.9) and eventually $Max_r$ contains exactly one side, leading to the gathering of all robots on one side. Otherwise, $Max_r$ contains all the sides. 
\end{proof}

\begin{remark}\label{four_conditions}
    A \texttt{MONITOR}-colored robot $r$ with no robot with color \texttt{OFF}, or \texttt{FINISH}, or \texttt{FINISH1}, or \texttt{FINISH2} eventually satisfies one of the following conditions.\\
    (i) $Max_r = \{S_r \}$ and $\Delta S$ contains no robot where $S\in \{ S_r^{opp}, S_r^L, S_r^R\}$ (Case 1.4.1, Section \ref{square})\\ (ii) $Max_r = \{S_r, S_r^{opp}\}$ and no robots on $\Delta S$ for $S \in \{S_r^L, S_r^R \}$ (Case 1.4.2, Section \ref{square}) \\ (iii) $Max_r = \{ S_r, S \}$ and both $\Delta S_r^{opp}$ and $\Delta S^{opp}$ contain no robot where $S \in \{S^L_r, S^R_r\}$ (Case 1.4.3, Section \ref{square})\\ (iv) $Max_r = \{ S_r, S_r^{opp}, S_r^L, S^R_r \}$ (Case 1.4.4, Section \ref{square}).
\end{remark}

\begin{lemma}\label{collision-free-square}
    The robot $r$ moving from the apex point to one of the sides of $\Re$ does not collide while moving. 
\end{lemma}

\begin{proof}
    Let $S_{Tgt}$ be the target of $r$, calculated based on $Max_r$. If $S_{Tgt} = S_r$, $r$ considers $\overleftrightarrow{rr_{Nbr_1}}$ to move back to $S_r$ where $r_{Nbr_1}$ is the neighbour of $r$ on $\mathcal{CH}_r$, not lying on $\Delta S_r$. By the similar argument presented in Lemma \ref{move_to apex point_square}, we can argue that the movement of the robot $r$ is free from collision. When $S_{Tgt} \neq S_r$, the movement of $r$ is similar to the strategy \textsc{Movement\_To\_LongestSide} which is collision-free by Lemma \ref{lemma3.1}.
\end{proof}

\begin{lemma}\label{monitor-finish-1}
    A \texttt{MONITOR}-colored robot $r$ on $(S_r, \frac{1}{8}S_r]$ satisfying one of the conditions in Remark \ref{four_conditions} terminates to a unique partition of $\Re$ and the movement is collision-free.
\end{lemma}
\begin{proof}
    If $r$ satisfies (i), it follows Type I partitioning, whereas if $r$ satisfies (ii), it follows Type II partitioning. The movement strategy for both cases is similar to the algorithm \textsc{Rectangular\_Partition} (in Section \ref{rectangle}). So, the statement of the above lemma holds from the Lemma \ref{apex_to_termination}. 
    If $r$ satisfies (iii), it follows Type III partitioning. It first calculates two points $A$ and $B$ on $S_r$, which are $\frac{len(S_r)}{c'}$ distance apart from $e_P$ and $e_Q$ respectively, where $c'$ is the number of robots on $\Delta S_r$ including itself and $e_P$ and $e_Q$ are the common endpoints of $S_r$, $S$ and $S_r$, $S^{opp}$, respectively. 
    It then chooses two triangles $\mathcal{T}_1 =\Delta Ae_Pe_R$ and $\mathcal{T}_2 = \Delta Be_Qe_R$ where $e_R$ is the vertex of the Type III partitioning. 
    If another robot $r'$ satisfies the conditions similar to those of $r$, it also chooses the above-mentioned two triangles. 
    The choice of the apex point for $r$ (as described in the strategy \textsc{MonitorRobot\_Movement\_Square} ensures that $S^{opp}$ must lie on either the same or different half-planes delimited by $L_r$ where other robots of $\Delta S_r$ lie. 
    If $S^{opp}$ lies on the half plane delimited by $L_r$ where other robots on $\Delta S_r$ lie, $r$ and $r'$ choose the centroid of $\mathcal{T}_1$ and $\mathcal{T}_2$ as their final positions, respectively. 
    Otherwise, $r$ and $r'$ choose the centroid of $\mathcal{T}_2$ and $\mathcal{T}_1$ as their final positions, respectively. 
    So, $r$ and $r'$ select different triangles for their termination and their two final positions are $\frac{2}{3}len(S_r) - \frac{2len(S_r)}{3c'}$ distance apart from each other, and the movement is free from any collision. Thus, $r$ and $r'$ both move to distinct partitions, each having $e_R$ as the vertex and with an area of $\frac{(len(S_r))^2}{2c'}$.
    If $r$ satisfies (iv), it follows Type IV partitioning. It calculates $A$ and $B$ on $S_r$ and chooses the centroid of one of the triangles $\Delta Ae_PC$ and $Be_QC$ as its final point. Based on similar arguments presented above, $r$ moves to a unique partition without collision. 
\end{proof}

\begin{lemma}\label{monitor-finish-2}
    A \texttt{MONITOR}-colored robot $r$ on $(S_r, \frac{1}{8}S_r]$ seeing either a \texttt{FINISH} or \texttt{FINISH1} or \texttt{FINISH2}-colored robot, terminates to a unique partition of $\Re$ and the movement is collision-free. 
\end{lemma}

\begin{proof}
    If $r$ sees a \texttt{FINISH}-colored robot $r''$ in $Int(\Re)$, it waits till $r''$ reach either on $\frac{1}{2}S_r$ or on $\frac{1}{4}S_r$ or on $\frac{3}{4}S_r$. Since $r''$ is in its move phase to its final position, $r$ eventually finds all the \texttt{FINISH}-colored robots lying on $\frac{1}{2}S_r$ or  $\frac{1}{4}S_r$ or  $\frac{3}{4}S_r$. Then it moves to its final position in a partition by following a similar strategy as an \texttt{OFF}-colored robot, described in Case 2 of the algorithm \textsc{Rectangle\_Partition} (in Section \ref{rectangle}). In this case, the statement of the above lemma follows from the similar arguments in Lemma \ref{off_to_terminate}.

    If $r$ sees a \texttt{FINISH1}-colored robot in $Int(\Re)$, it waits for all \texttt{FINISH1}-colored robots to reach on either $\frac{1}{3}S_r$ or $\frac{1}{3}S$ where $S \in \{ S_r^L, S_r^R \}$ and eventually sees all those robots reach their final positions.  Since the \texttt{FINISH1}-colored robot on $\frac{1}{3}S_r$ may reach to another side triangle than $\Delta S_r$ while executing its movement to the final position, we separately count $c'$ (which is the number of the robot on $\Delta S_r$ without color \texttt{FINISH1}) and $c$ (which is the number of \texttt{FINISH1}-colored robots on $\frac{1}{3}S_r$). Similar to Lemma \ref{monitor-finish-1}, $r$ selects $A$, $B$ which are $\frac{len(S_r)}{c+c'}$ distance apart from $e_P$ and $e_Q$, respectively and the centroid of one of the two triangles $\Delta Ae_Pe_R$ and $\Delta B e_Q e_R$ as its final position. By Lemma \ref{monitor-finish-1}, this final position of $r$ lies in a unique partition of $\Re$ and the movement is collision-free. Similarly, we can prove that if $r$ sees a \texttt{FINISH2}-colored robot, it terminates to a unique partition without collision. 
\end{proof}

\begin{lemma}\label{off-finish}
    An \texttt{OFF}-colored robot $r$ on $S_r$  terminates to a unique partition of $\Re$, and the movement is collision-free. 
\end{lemma}

\begin{proof}
  By Remark \ref{four_conditions}, there exists a terminal robot $r''$ on one of the sides, which moves to its apex point on $(S_r, \frac{1}{8}S_r]$ with color \texttt{MONITOR} and satisfies one of the four conditions as described in the remark. By Lemma \ref{monitor-finish-1} and \ref{monitor-finish-2}, $r''$ reaches to a unique partition of $\Re$ with color either \texttt{FINISH} or \texttt{FINISH1} or \texttt{FINISH2}. Let us assume that $r$ is a terminal robot of $S_r$. If $r''.color = $ \texttt{FINISH}, $r$ follows a similar strategy as presented in Case 2 of \textsc{Rectangle\_Partition}. By Lemma \ref{off_to_terminate}, $r$ terminates to a unique partition without collision. 

    Let $r''.color =$ \texttt{FINISH1}. So, $r$ understands that it needs to follow Type III partitioning. In this case, $r$ must have seen some robots on $\Delta S$ ($S \neq S_r$) and calculates the distance $bl$, the two points $e_Y$ and $e_Z$ and the number $c$, as described in Case 3 (Section \ref{square}). It then selects two points $e_Y$ and $e_Z$ and moves to the centroid of the triangle $\Delta e_Y e_R e_Z$. If $r'$ is another terminal robot on $S_r$, which simultaneously gets active with $r$, then it selects the centroid of the different triangle than that of $r$, as described in Case 3 (in Section \ref{square}). The choice of the triangles ensures that the two robot $r$ and $r'$ terminate at $t_r$ and $t_{r'}$. Let $u_1$ be the point of intersection of $\frac{1}{3}S_r$ and $D_r$ and $u_2$ be the point of intersection of $\frac{1}{3}S_r$ and $S^{opp}$. If we consider that $k$ (and $k'$) is the number of \texttt{FINISH1}-colored robots on $\frac{1}{3}S_r$ starting from the robot $\frac{bl}{3}$ distance apart from $u_1$ ($u_2$) towards $u_2$ ($u_1$) such that two consecutive robots are $\frac{2bl}{3}$ distance apart from each other, where $d(t_r, t_{r'}) = \frac{2}{3}len(S_r) - \frac{2}{3}(k+k'+1)bl > 0$ as $len(S_r) \geq (k+k')bl + 2bl$. Hence, $r$ and $r'$ terminate at distinct partitions of the region $\Re$. The movement of the two robots is free from collision, by the similar arguments presented in Lemma \ref{apex_to_termination}.

    For $r''.color = $ \texttt{FINISH2}, we can similarly prove that $r$ terminates at a unique partition.
    \end{proof}

\begin{theorem} \label{timecomplexity_square}
    Algorithm \textsc{Square\_Partition} solves uniform partitioning for the square region in $O(N)$ epochs without collision.
\end{theorem}

\begin{proof}
    By Lemma \ref{lemma4.1}, starting from any initial configuration, all the interior and corner robots become boundary robots in $\Re$. After this, all the robots eventually form either a one-side, two-side, or four-side configuration by Lemma \ref{lemma4.2}. Thereafter, monitor robots terminate to distinct partitions of $\Re$ by Lemma \ref{monitor-finish-1}, \ref{monitor-finish-2} and all the other robots terminate to distinct partitions by Lemma \ref{off_to_terminate}. Hence, \textsc{Square\_Partition} solves uniform partitioning for the square region.

    In the worst case, all robots may lie on a straight line on $Int(\Re)$, in which it takes $O(N)$ epochs for all robots to reach the boundary. After moving to the boundary, it takes one epoch for monitor robots to move to apex points. In the worst case, there may be 8 \texttt{MONITOR}-colored robots in $Int(\Re)$, each lying on a different line parallel to $S_r$ and $r$ is the farthest robot from $S_r^{opp}$. If $Max_r = \{S_r^{opp}\}$, $r$ needs to wait for seven epochs to reach $S_r^{opp}$ because the other seven robots move to $S_r^{opp}$ one by one. Then $r$ needs one more epoch to reach $S_r^{opp}$. Thus, $r$ takes nine ($O(1)$) epochs so that the whole configuration of the robots becomes a one-side configuration. Hence, reaching a one-side configuration takes $O(N)$ epoch. Similarly, we can prove that if $Max_r$ contains more than one element, we can reach a two-side or a four-side configuration. After this, if $r$ is again a monitor robot, it needs one epoch to reach its apex point, waits for another seven epochs for other seven \texttt{MONITOR}-colored robots to reach their respective final positions, and needs one more epoch to reach its final position in a partition. Thus, $r$ terminates in nine epochs. After seeing $r$ at its final position, all the other \texttt{OFF}-colored robots on the sides position themselves at their respective final positions one by one in $O(N)$ epochs, as each of them needs $O(1)$ epoch to reach their final position. Hence, \textsc{Square\_Partition} solves the uniform partitioning problem in $O(N)$ epochs. 
    

    The algorithm is free from collision by Lemma \ref{lemma4.1}, \ref{move_to apex point_square}, \ref{collision-free-square}, \ref{monitor-finish-1}, \ref{monitor-finish-2}, and \ref{off-finish}.
\end{proof}

\section{Algorithm for a Circular Region}
\label{circle}
This section considers the region $\Re$ as a circle of radius $rad$.
We follow the Type V partitioning for circular regions. Our target in this process is to form a regular $N$-gon inscribed in the region using the robots. In this section, we discuss two techniques to achieve the target.

We can use an algorithm for the widely popular \emph{Uniform Circle Formation} (UCF) problem, presented by Feletti et al \cite{feletti_et_al:LIPIcs.OPODIS.2023.5}. They proposed a $O(\log(N))$ algorithm for the ASYNC setting. The algorithm relies on the mutual visibility algorithm, proposed by Sharma et al \cite{a14020056}, where the robots first position themselves on the vertices of a convex hull. Taking advantage of the mutual visibility, robots then place themselves on the smallest enclosing circle, which later gets translated to a uniform circle configuration using a subroutine, named as \texttt{Uniform Transformation} that runs in $O(\log(N))$ epochs. The mutual visibility algorithm in \cite{a14020056} requires $47$ colors to achieve the goal in $O(1)$ epochs. Regarding the uniform partitioning problem in a circular region, we can achieve the same goal of repositioning the robots on the boundary of $\Re$ using no colors in $O(1)$ epochs. From the application-oriented point of view, we can take advantage of the bounded region where the robots can detect the boundary of the region. Robots can be placed on the boundary using the movement strategy \textsc{Move\_To\_Boundary} (which is described later in this section), maintaining the initial color \texttt{OFF} throughout the strategy.  
The detailed steps of the technique \texttt{Uniform Transformation} are not repeated here to keep the section concise. We include the pseudocode of the algorithm in \ref{appendix}.

\subsection{Description of the Algorithm}
Initially, all robots are with color \texttt{OFF}. 
Before going into the details of the algorithm, we define the following notations that we will use in the future to comprehend the algorithm better. 

\noindent \textbf{Notations:} For a robot $r$, 
\begin{itemize}

    \item $O$ is the center of $\Re$. $p_r$ and $p_r^{opp}$ are the two diametrically opposite points of intersection of the line $\overleftrightarrow{rO}$ and the boundary of $\Re$, out of which $p_r$ is the nearest to $r$.
    
    \item  $\wideparen{AB}$ represents the segment of the circumference of $\Re$, starting from the point $A$ to $B$ such that the segment has no robot except on the endpoints. In the case of $\wideparen{AB}$ having no robots in both clockwise and counter-clockwise directions, $\wideparen{AB}$ is chosen to be the arc with the shortest arc length.
    
    \item $alen(AB)$ is the arc length of $\wideparen{AB}$.
    
\end{itemize}

\begin{wrapfigure}[12]{r}{0.4\textwidth}
\centering
     \includegraphics[width=\linewidth]{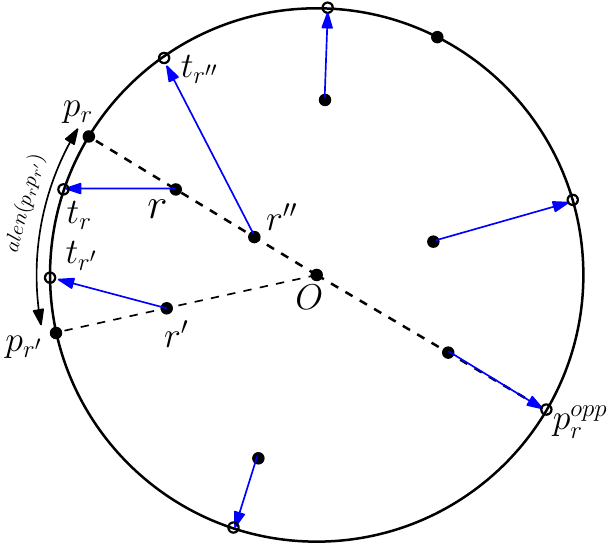}
     \caption{The target point $t_r$ of $r$ depends on the distance of $r$ from the point $p_r$}
     \label{fig.circle_moving_to_boundary}
\end{wrapfigure}
\noindent \textbf{Strategy} \textsc{Move\_To\_Boundary}: If $r$ is a boundary robot and there is at least one  \texttt{OFF}-colored interior robot visible to it, $r$ does not move or change its current color. If $r$ is an interior robot and lies on the center $O$, it waits till all other interior robots with color \texttt{OFF} reach the boundary of $\Re$. 
When there are no other interior robots left, $r$ does not change its color and moves to any target point $t_r$ on the boundary that does not contain a robot. Otherwise ($r$ is an interior robot not lying on $O$), $r$ considers the point $p_r$ and checks whether the point $p_r$ is visible or not. 
If $p_r$ is visible and no robot lies on it, $r$ simply moves to $p_r$ with current color \texttt{OFF}. 
On the other hand, $r$ computes the set $\mathcal{V}_r$, that consists of all visible robots to $r$ not lying on the line $\overleftrightarrow{rO}$. We identify two sub-cases.

\begin{itemize}
    \item \textbf{$\mathcal{V}_r$ is non-empty:} $r$ calculates $d_r = \frac{1}{4}\min \{ alen(p_rp_{r'}) |$ $ r' \in \mathcal{V}_r \text{ and } \arc{p_rp_{r'}} \text{ is defined} \}$, as shown in Fig. \ref{fig.circle_moving_to_boundary}. 

    \item \textbf{$\mathcal{V}_r$ is empty:} It happens when all the robots lie on one line passing through $O$. Here, $r$ calculates $d_r = \frac{1}{4} alen(p_rp_r^{opp})$.

\end{itemize}

Finally, $r$ finds the target point $t_r$ such that $alen(t_r p_r) = \frac{d_r}{rad} d(r,p_r)$ and moves to $t_r$ with the current color \texttt{OFF}, as illustrated in Fig \ref{fig.circle_moving_to_boundary}. 

Observe that the above strategy does not require any change in color to displace the robots on the boundary from the interior of $\Re$. From now on, we can use the \texttt{Uniform Transformation} sub-problem proposed in \cite{feletti_et_al:LIPIcs.OPODIS.2023.5} to achieve a configuration of the robots such that $alen(r_ir_j) = \frac{2\pi \cdot rad}{N}$ for any two neighbouring robots $r_i$ and $r_j$ in $O(\log (N))$ epochs under ASYNC setting. The above technique uses more than $17$ colours to achieve a uniform circular configuration.

Although the number of colors required in \texttt{Uniform Transformation} is a constant, if we turn our focus on reducing the number of colors, we can propose another technique to achieve the same goal of repositioning the robots in such a way that the neighbouring robots are $\frac{2\pi \cdot rad}{N}$ distance apart from each other.
Observe that when all robots are on the boundary of $\Re$, they all can see each other, as no three robots are collinear. We will turn this to our advantage. 
We now define a \emph{cluster} that will be useful for the rest of the algorithm.

\begin{definition}{(Cluster)}
    A sequence of robots $Cl = \{r_1, r_2, \cdots, r_k \}$ ($k \geq 1$) lying on the boundary of $\Re$, is called a cluster if the following conditions hold. (i) $r_i$ and $r_{i+1}$ are consecutive robots on the boundary. (ii) $alen(r_ir_{i+1}) = \frac{2\pi \cdot rad}{N}$. (iii) $alen(r_1r_1^{Nbr}), alen(r_kr_k^{Nbr}) \neq \frac{2\pi\cdot rad}{N}$, where $r_1^{Nbr}$ and $r_k^{Nbr}$ are the neighbors of $r_1$ and $r_k$ respectively such that $r_1^{Nbr}, r_k^{Nbr} \notin Cl$. 
\end{definition}

\begin{wrapfigure}[12]{r}{0.4\textwidth}
\centering
     \includegraphics[width=0.9\linewidth]{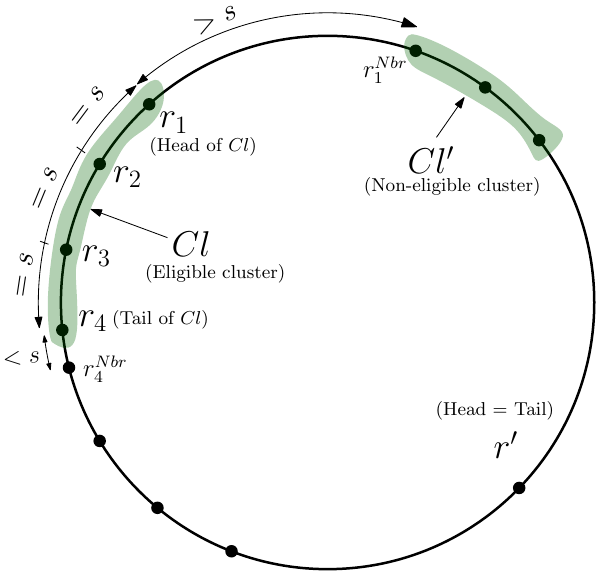}
     \caption{Illustrating eligible clusters with head and tail}
     \label{fig.circle_cluster}
\end{wrapfigure}
For convenience, we denote ${(2\pi \cdot rad)}/N$ by $s$. Any robot can detect its own cluster, as it is able to see all the robots on the boundary. A robot $r$ is called the \emph{head} of a cluster $Cl$, if $alen(rr^{Nbr_1}) = s$ and $alen(rr^{Nbr_2}) > s$ for two neighbours $r^{Nbr_1}, r^{Nbr_2}$ of $r$ such that  $r^{Nbr_1} \in Cl $ and $r^{Nbr_2} \notin Cl $. The robot $r$ is called \emph{tail} of the cluster $Cl$, when $alen(rr^{Nbr_1}) = s$ for $r^{Nbr_1} \in Cl $ and $alen(rr^{Nbr_2}) < s$ for $r^{Nbr_2} \notin Cl $. It is possible that the head or tail of a cluster does not exist. Moreover, a cluster can be made of a single robot $r$. In such cases, if $alen(rr^{Nbr_1}) > s$ and $alen(rr^{Nbr_2}) < s$, $r$ itself is called the head and tail of the cluster. An example is given in Fig. \ref{fig.circle_cluster}. Note that a cluster having no head and tail contains all the robots in $\Re$.

\begin{definition}{(Eligible Cluster)}
    A cluster is eligible for movement if it satisfies the following two conditions.
    (i) The cluster has both head and tail robots.
        (ii) All the robots in that cluster are with color \texttt{OFF}.
\end{definition}
 In other words, clusters that only have heads or tails are not eligible for movement. In an eligible cluster, the head initiates the movement along the perimeter of $\Re$. By movement of a cluster $Cl = \{r_1,r_2, \cdots, r_k \}$, we mean that all the robots in $Cl$ move in the direction of the head sequentially. Note that the robots do not have identifiers. We use $r_1, r_2, \cdots, r_k$ for better comprehension. The movement of a cluster is explained below.  

\noindent \textbf{Movement of an Cluster $Cl$:}  Let us consider that $r_1$ is the head, $r_k$ is the tail of $Cl$ and $r'_1$ is neighbor of $r_1$, not in $Cl$. Let $Cl'$ be the cluster of $r'_1$, as shown in Fig. \ref{fig.circle_cluster_movement_case1}.
When $r_1$ gets activated with color \texttt{OFF} with all other robots in $Cl$ having color \texttt{OFF}, and finds $Cl$ as an eligible cluster, it changes its color to \texttt{HEAD} and waits till all robots in its cluster change their color either to \texttt{MID} or \texttt{TAIL}. 
When $r_k$ finds $r_1$ in its own cluster $Cl$ with color \texttt{HEAD} and determines that all the other robots in  $Cl$ lies between $r_1$ and $r_k$ on the boundary of $\Re$, it sets its color to \texttt{TAIL}. All other robots in $Cl$ change their color to \texttt{MID} after seeing $r_1$ with the color \texttt{HEAD}. 
After this, $r_1$ determines the eligibility of $Cl'$. 
If any robot lying on the boundary sees another interior robot with color other than \texttt{OFF}, it waits as the robots are in their move phase toward their target points through the interior of $\Re$.
Otherwise, we have two cases.

\begin{figure}
\begin{minipage}[c]{0.32\textwidth}
     \centering
     \includegraphics[width=\linewidth]{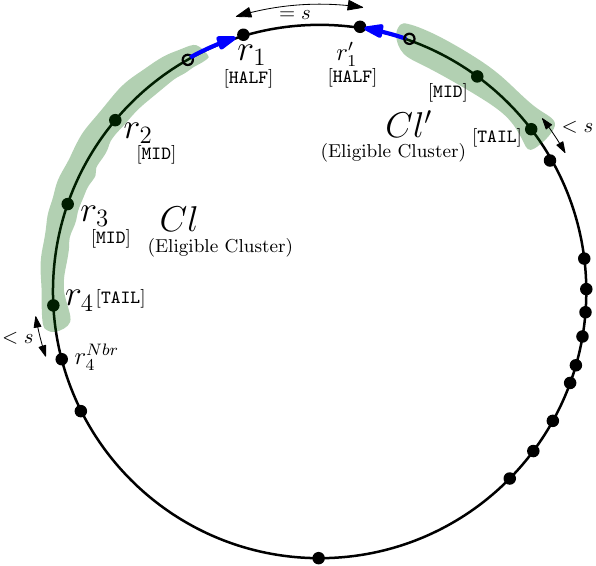}
     \caption{Movement of the heads of the eligible clusters}
     \label{fig.circle_cluster_movement_case1}
 \end{minipage}
 \hfill
 \begin{minipage}[c]{0.32\textwidth}
     \centering
     \includegraphics[width=\linewidth]{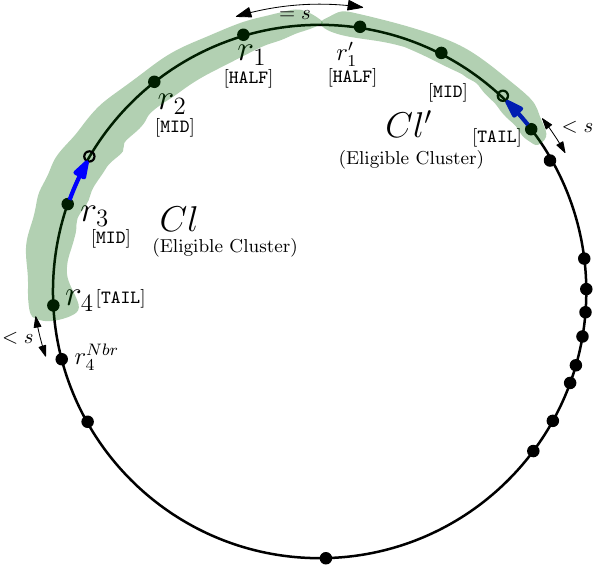}
     \caption{Movements of the robots with color \texttt{MID} \& \texttt{TAIL}}
     \label{fig.circle_cluster_movement_case1_1}
 \end{minipage}
 \hfill
 \begin{minipage}[c]{0.32\textwidth}
     \centering
     \includegraphics[width=0.95\linewidth]{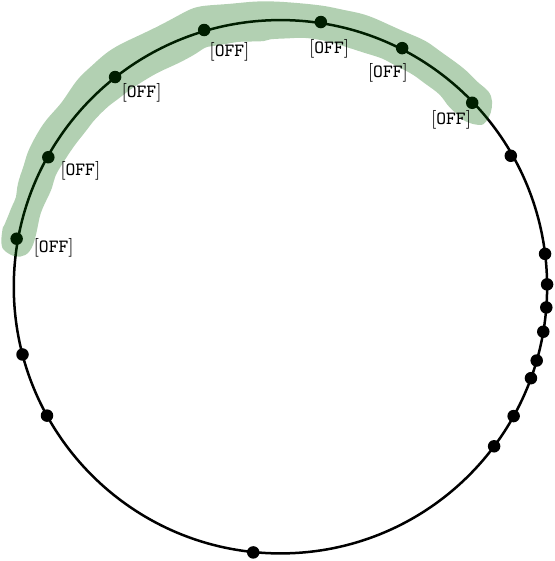}
     \caption{Two clusters are merged with each other}
     \label{fig.circle_cluster_movement_case1_2}
 \end{minipage}
 \end{figure}

\noindent \textbf{Case 1 ($Cl'$ is eligible and $r'_{1}.color =$ \texttt{OFF}):} In this case, our strategy is to move each head of $Cl$ and $Cl'$ a distance $\frac{1}{2}(alen(r_1r'_1)-s)$ towards each other so that the distance between them becomes $s$ after movement. All other robots of the respective clusters sequentially move in the direction of the head's movement. The process is as follows and illustrated in Fig. \ref{fig.circle_cluster_movement_case1}. 
The head $r_1$ changes its color to \texttt{MOVE-H} and moves towards $r'_1$ to a point $t_{r_1}$ such that $alen(t_{r_1}r_1) = \frac{1}{2}(alen(r_1r'_1)-s)$. After getting activated with color \texttt{MOVE-H}, it changes its color to \texttt{HALF}.
It waits till the tail of the cluster (if it exists) $r_k$ sets its color to \texttt{OFF}. 
At this moment, $r_k$ and all \texttt{MID}-colored robots do not change their position or color. 
After seeing $r_1$ on the boundary of $\Re$ with the color \texttt{HALF} and $alen(r_1, r_2) > s$, $r_2$ starts moving towards $r_1$ without changing its current color to a point $t_{r_2}$ on the boundary such that $alen(t_{r_2}r_2) = alen(r_1r_2)-s$.
Other robots in $Cl$ except $r_2$ remain in place at this time. 
Similarly, for any robot $r_i$ ($3 \leq i \leq k$), if $r_{i-1}$ lies on the boundary of $\Re$ with $alen(r_{i-1}r_i) > s$ and $alen(r_{i-1}r_{i-2}) = s$, $r_i$ moves to $t_{r_i}$ towards $r_{i-1}$ with its current color such that $alen(t_{r_i}r_i) = alen(r_{i-1}r_i)-s$, as shown in Fig. \ref{fig.circle_cluster_movement_case1_1}. 
After the movement, $r_k$ finds the cluster $Cl$ with its head $r_1$ having color  \texttt{HALF}.
It changes its color to \texttt{OFF}. The \texttt{MID}-colored robots in $Cl$ change their color to \texttt{OFF} after seeing the tail with color \texttt{OFF}.
The head $r_1$ changes its current color to \texttt{OFF} from \texttt{HALF} as illustrated Fig. \ref{fig.circle_cluster_movement_case1_2}, when both of its neighbours are $s$ arc-length away from it and all the robots in its cluster are with color \texttt{OFF}.

\begin{figure}[H]
 \begin{minipage}[c]{0.47\textwidth}
     \centering
     \includegraphics[width=0.8\linewidth]{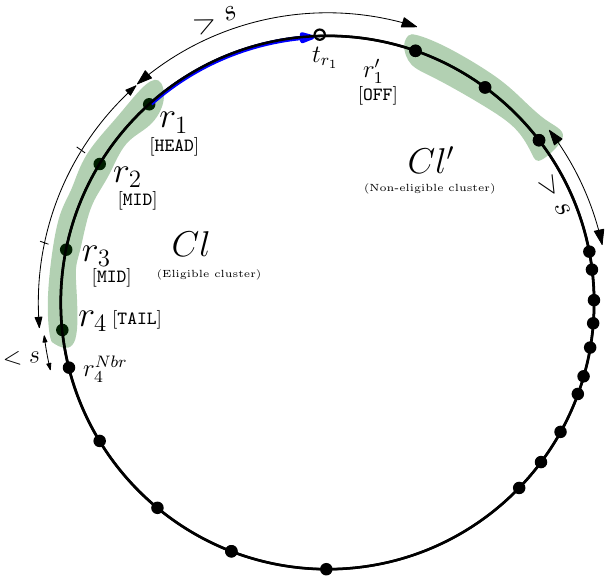}
     \caption{$Cl$ is an eligible cluster and $Cl'$ is the non-eligible cluster}
     \label{fig.circle_cluster_movement_case2}
 \end{minipage}
 \hspace{0.04mm}
 \begin{minipage}[c]{0.47\textwidth}
     \centering
     \includegraphics[width=0.72\linewidth]{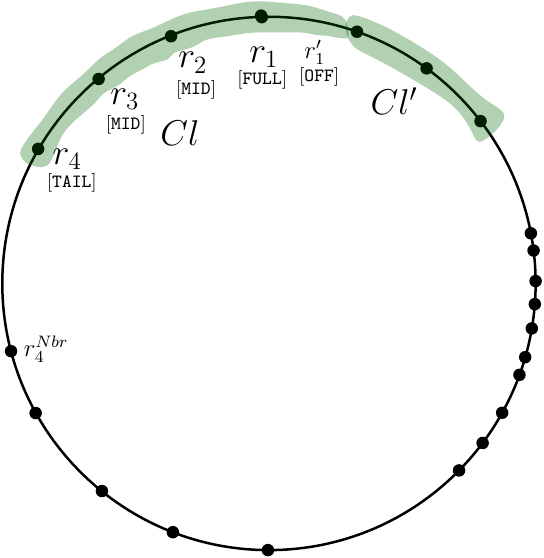}
     \caption{The two clusters got merged after the movement of $Cl$ }
     \label{fig.circle_cluster_movement_case2_1}
 \end{minipage}
 \end{figure}

\noindent \textbf{Case 2 ($Cl'$ is eligible and $r'_1.color = $ \texttt{HALF}) or ($Cl'$ is non-eligible):} Here, we move the head $r_1$ with color \texttt{FULL} towards $r'_1$ to a point on the boundary $t_{r_1}$ such that $alen(t_{r_1}r_1) = (alen(r_1r'_1)-s)$, as shown in Fig. \ref{fig.circle_cluster_movement_case2}. 
The rest of the strategy is the same as the previous case (Case 1) for other robots in the cluster. Finally, $Cl$ gets merged with $Cl'$, as shown in Fig. \ref{fig.circle_cluster_movement_case2_1}. When $r_1$ is with color \texttt{FULL} and finds all the robots in $Cl$ with color \texttt{OFF}, it sets its color to \texttt{OFF}.

When any robot $r$ finds itself with color \texttt{OFF} and all robots are in one cluster, i.e., every two consecutive robots on the boundary are at $s$ distance apart from each other, it moves to the point $t_r$ on the line segment $\overline{rO}$ such that $d(r,t_r) = \frac{1}{2}d(r,O)$ with color \texttt{FINISH}. Even if a boundary robot $r$ sees a \texttt{FINISH}-colored robot in $Int(\Re)$, it follows the same strategy. In all other cases, $r$ maintains status-quo. 

\subsection{Analysis of the Algorithm }
In this subsection, we discuss the correctness and the time complexity of the above-mentioned algorithm for circular regions. We also prove that the robots do not meet collision during any movement.

\begin{lemma}
    \label{lemma:5.1}
    Any interior robot in $\Re$ with color \texttt{OFF} moves to its boundary without collision.
\end{lemma}
 \begin{proof}
Let us consider two robots, $r$ and $r'$, lying in $Int(\Re)$. We can identify two cases here.

\noindent \textbf{Case 1 ($r'$ lies on $\overleftrightarrow{rO}$):} In this case, we investigate two sub-cases. When $p_r = p_{r'}$, it signifies that $r$ and $r'$ lie  on $\overline{Op_r}$.
Without loss of generality, let us assume that $r$ is farther from $p_r$ than $r'$. 
So, $d(p_r, r') < d(p_r, r)$. We further assume that $p_r$ has a robot on it. If both of them get activated simultaneously, they find each other in the same radius. 
Both of them calculate $d_r$ and $d_{r'}$ respectively, such that $d_r = d_{r'}$. If they choose $t_r$ and $t_{r'}$ as their respective target points on the boundary, where $alen(t_rp_r) = \frac{d_r}{rad}d(r,p_r) > \frac{d_{r'}}{rad}d(r',p_{r'}) = alen (t_{r'}p_{r'})=alen(t_{r'},p_r)$. We also have $d(r,p_r) > d(r', p_r)$. Both of the above inequalities imply that the line passing through $r_{mid}$ (the midpoint of $\overline{rr'}$) and $t_{r_{mid}}$ (the midpoint of $\wideparen{t_rt_{r'}}$) separates the paths of $r$ and $r'$ towards their respective target points. 
In case of $p_r \neq p_{r'}$, if we assume that $\mathcal{V}_r$ is non-empty with $r'' \in \mathcal{V}_r$, then the line $\overleftrightarrow{r''O}$ separates the two target points $t_r$ and $t_{r''}$. If $\mathcal{V}_r$ is empty, the line $\overleftrightarrow{p_{r_{mid}}O}$ separates $\overline{rt_r}$ and $\overline{r't_{r'}}$ where $p_{r_{mid}}$ is the midpoint of $\wideparen{p_rp_{r'}}$. So, the movement paths of the two robots $r$ and $r'$ do not cross each other. Hence, the movements of $r$ and $r'$ are free from collision. 

\noindent \textbf{Case 2 ($r'$ does not lie on $\overleftrightarrow{rO}$):}
In this case, $p_r \neq p_{r'}$. By the similar argument presented above, the movements of the robots are free from collision.
\end{proof}

\begin{lemma}
    \label{lemma:5.2}
    There always exists an eligible cluster or all the robots are in the same cluster.
\end{lemma}
\begin{proof}
    If all the robots are in the same cluster, the statement of the lemma follows trivially. Let us assume that there are at least two clusters. If all the clusters are non-eligible, then all the arc-length between two consecutive clusters is always either greater or lesser than $s$, which is a contradiction to the fact that the sum of the distances between two consecutive robots on the boundary is $2\pi \cdot rad = s \cdot N$.
\end{proof}

\begin{lemma}
    \label{lemma:5.2.1}
    After all the robots in a cluster complete their movement, they remain as a cluster.
\end{lemma}

\begin{proof}
    For an eligible cluster $Cl = \{ r_1, r_2, \cdots r_k \}$, the head $r_1$ executes its movement towards another cluster. After that, $alen(r_1r_2) > s$, since $r_2$ is a member of $Cl$. When $r_1$ reaches its final position and changes its color to either \texttt{HALF} or \texttt{FULL}, $r_2$ moves to a point $t_{r_2}$ towards $r_1$ such that $alen(r_1t_{r_2}) = s$. All other robots in $Cl$ sequentially execute the same process which leads to the completion of one movement of the whole cluster. Hence the statement follows.
\end{proof}

\begin{lemma}
    \label{lemma:5.3}
    Let $k ~(<N)$ be the length of an eligible cluster $Cl$. The length of $Cl$ gets increased at least by one without collision in $O(k)$ epochs.
\end{lemma}

\begin{proof}
    First the head of an eligible cluster having all robots with color \texttt{OFF}, changes its color to \texttt{HEAD}. Upon seeing this, all other robots in the cluster change their color to  either \texttt{MID} or \texttt{TAIL}. This operation takes total $2$ epochs. For a cluster $Cl = \{ r_1, r_2,\cdots r_k\}$ with $r_1$ as head and $r_k$ as tail, $r_1$ moves first in the direction of the another cluster. Then \texttt{MID}-colored robots move sequentially (first $r_2$ moves in the direction of $r_1$, then $r_3$ and at last the tail $r_k$). This step takes $O(k)$ epochs. If the cluster $Cl$ moves towards another cluster $Cl'$, the heads of the two clusters become $s$ arc-length apart in just one epoch (either both the heads move towards each other or one head moves towards the other). After the movement of the whole cluster, the tail of that cluster changes its color to \texttt{OFF}, then the \texttt{MID}-colored robots change their color to \texttt{OFF} and finally the head changes its color to \texttt{OFF} when both of its neighbors are $s$ arc-length apart from it and all other robots in its cluster are with color \texttt{OFF}. This process takes $3$ epochs. So, the movement of a cluster of length $k$ takes $O(k)$ epochs. After the movement, both the neighbor of $r_1$ are $s$ distance apart from each other. Therefore the length of the cluster gets increased at least by one.
    Hence the proof.
\end{proof}

\begin{theorem}
\label{timecomplexity_circle}
    Our algorithm solves uniform partitioning for the circular region in $O(N^2)$ epochs.
\end{theorem}

\begin{proof}
    The robots lying on $Int(\Re)$, moves to the boundary.  The process of moving all the interior robots to the boundary takes $O(1)$ epochs. In the worst case, it is possible that the length of an eligible clusters gets increased only by one in every movement of the cluster. Lemma \ref{lemma:5.3} proves that the movement of an eligible cluster takes $O(k)$ epochs, where $k$ is the length of the cluster.  So, it takes $\sum_{k=1}^{N-1} O(k) \approx O(N^2)$ epochs for all robots to become one cluster. Finally, all robots moves to their final positions in one epochs. Thus, our algorithm takes overall $O(N^2)$ epochs to partition the region.   
\end{proof}


\section{Conclusion} \label{conclusion}
We studied the distributed version of  uniform partitioning of a bounded region using mobile robots. The problem becomes interesting and challenging due to the robot model that is considered in this paper. The robots are opaque and do not have enough memory to store the past information. They have a persistent memory in form of a light. The ASYNC activation schedule which is considered in this paper, is the most general form of activation of a robot network, where any robot can be activated any time. Moreover, the robots do not have any coordinate axes agreement or global orientation. We solved this problem when the region is either a rectangle, a square or a circle. The reason behind this, is the application oriented point of view, as the regions, we deal with in our daily life, are known geometric figures. We believe that the problem can be extended to other convex regions under the same model. Also, the problem can even be studied when some of the robots become faulty.

\bibliography{ICDCN_24.bib}

\newpage
\appendix

\section{Pseudocode of the Algorithms}\label{appendix}

\textsc{Rectangle\_Partition} is the main pseudocode of the algorithm for rectangular regions. \textsc{Square\_Partition} is the main pseudocode of the algorithm for square regions. \textsc{Circle\_Partition} is the main pseudocode of the algorithm for circular regions. 

  \begin{algorithm2e}[!ht]
	\If{$r.color =$ \texttt{OFF} }
        {
            Follow \textsc{Rectangle\_Target\_Side}() to calculate $S_r$\\
            
            \If{$r$ is not a boundary robot on $S_r$}
                {
                    Follow \textsc{TargetPoint\_On\_Sides}() to find a target point \\
                    Move to the target point with the current color \texttt{OFF}
                }
            \Else(\tcp*[h]{$r$ lies on one of the longest sides of $\Re$})
                {
            
                    \If{$r$ is a monitor robot on $S_r$}{
                        \If{Each of $S_r$ and $S_r^{opp}$ contains only one robot}{
                            Change $r.color$ to \texttt{FINISH}\\
                            Move to the midpoint of $\frac{1}{4}S_r$
                        }
                        \ElseIf{$S_r^{opp}$ has exactly one robot, but $S_r$ has more than one robot}{
                            No change in color and position
                        }
                        \ElseIf{$S_r$ has exactly one robot, but $S_r^{opp}$ has more than one robot}{
                            Set $S_r^{opp}$ as the target side and follow \textsc{TargetPoint\_On\_Sides}() to find the target point\\
                            Move to the target point with the current color
                        }
                        \Else{
            
                            Follow \textsc{Move\_To\_ApexPoint}() and calculate the apex point $a_r$\\
                            Change $r.color$ to \texttt{FINISH} from \texttt{OFF} \\
                            Move to the point $a_r$
                        }}
                    \ElseIf{$r$ is a terminal robot on $S_r$ with a \texttt{FINISH}-colored robot on $\frac{1}{2}S_r$ or $\frac{1}{4}S_r$ or $\frac{1}{4}S_r^{opp}$}
                        {

                            Follow \textsc{Move\_To\_FinalPoint}() to compute the target point \\
                            Change $r.color$ to \texttt{FINISH} \\ Move to the target point
                        }
                    \Else
                        {
                            No change in color and position
                        }
                }
            
        }

	\caption{\textsc{Rectangle\_Partition}}
 \label{rectanglepartition}
\end{algorithm2e}

\begin{algorithm2e}
\setcounter{AlgoLine}{25}

\ElseIf{$r.color = $ \texttt{FINISH} and $r$ lies on $\frac{1}{8}S_r$}
                {
                    \If{there is an \texttt{OFF}-colored robot visible in $Int(\Re)$}
                        {
                            $r$ does not move and change its current color
                        }
                    \Else
                        {
                            \If{$r$ finds all the \texttt{FINISH}-colored robot on $\frac{1}{8}S_r$ and $\frac{1}{8}S_r^{opp}$}
                                {
                                    $c_r^1 \longleftarrow$ Number of robots on $[S_r, \frac{1}{8}S_r]$ including itself \\
                                    $c_r^2 \longleftarrow$ Number of robots on  $[S_r^{opp},\frac{1}{4}S_r^{opp}]$ 
                                    
                                    \If{$0 < c_r^2 < c_r^1$}
                                        {
                                            $S_r$ is the target side and follow \textsc{TargetPoint\_On\_Sides}() to find a target point  \\ 
                                            Move to the target point with the color \texttt{FINISH}
                                        }
                                    \ElseIf{$0 < c_r^1 < c_r^2$}
                                        {
                                            \If{there is a \texttt{FINISH}-colored robot on $[S_r^{opp}, \frac{1}{8}S_r^{opp}]$}
                                                {
                                                    $r$ does not change its color or position
                                                }
                                            \Else
                                                {
                                                    $S_r^{opp}$ is target side and follow \textsc{TargetPoint\_On\_Sides}() to find the target point \\
                                                    Change its color to $\texttt{OFF}$ and move to the target point
                                                }
                                        }
                        
                                    \Else(\tcp*[h]{$0 < c_r^1 = c_r^2$ or $0 = c_r^2 < c_r^1$})
                                        {
                                            Follow \textsc{Move\_To\_FinalPoint}() to find a target point \\
                                            Move to the target point with color \texttt{FINISH}
                                        }
                                }
                            \ElseIf{a \texttt{FINISH}-colored robot on $(\frac{1}{8}S_r, \frac{7}{8}S_r)$}
                                {
                                    $S_r$ is the target side and follow \textsc{TargetPoint\_On\_Sides}() to find a target point  \\ 
                                    
                                    Move to the target point with the color \texttt{FINISH}
                                }
                            \Else
                                {
                                    No change in color and position
                                }
                        }
                }
        
\ElseIf{$r.color = $ \texttt{FINISH} and $r$ lies on $S_r$}
            {
                Change $r.color$ to \texttt{OFF} without any movement
            }
\Else(\tcp*[h]{$r.color = $ \texttt{FINSH} and $r$ lies on $\frac{1}{2}S_r$ or $\frac{1}{4}S_r$})
            {
                $r$ terminates.
            }
\end{algorithm2e}

\begin{algorithm2e}[!ht]
    $C \longleftarrow$ The center of $\Re$ \\
    $e_P \longleftarrow$ The common corner of $S_r$ and $S_r^L$, \\ 
    $e_Q\longleftarrow$ The common corner of $S_r$ and $S_r^R$\\
    $e_R\longleftarrow$ The common corner of $S_r^{opp}$ and $S_r^R$ \\ 
    $e_X\longleftarrow$ The common corner of $S_r^{opp}$ and $S_r^L$\\
    
	\If{$r.color =$ \texttt{OFF} }
        {
            Choose one of the nearest sides of $\Re$ as $S_r$
            
            \If{$r$ is not a boundary robot on $S_r$}
                {
                    Follow \textsc{TargetPoint\_On\_Sides}() to find the target point \\
                    Move to the target point on $S_r$ with color \texttt{OFF}
                }
            \Else
                {
                    \If{$r$ is a monitor robot and $Int(\Re)$ has no \texttt{FINISH} or \texttt{FINISH1} or \texttt{FINISH2}-colored robot}
                        {
                            Follow \textsc{Apex\_Point\_inSquare}() to find the target point $t_r$ \\
                            Change $r.color$ to \texttt{MONITOR}\\
                            Move to the point $t_r$ 
                        }
                    \ElseIf{$r$ is a terminal robot and $Int(\Re)$ has \texttt{FINISH}-colored robots}
                        {
                            Follow \textsc{Move\_To\_FinalPoint}() to find a target point\\
                            Change $r.color$ to \texttt{FINISH}\\
                            Move to the target point
                            
                        }
                    \ElseIf{$r$ is a terminal robot and $Int(\Re)$ has \texttt{FINISH1}-colored robots}
                        {
                            Follow \textsc{PartitionType\_III} to find a target point \\
                            Change $r.color$ to \texttt{FINISH1} \\
                            Move to the target point
                        }
                    \ElseIf{$r$ is a terminal robot and $Int(\Re)$ has \texttt{FINISH2}-colored robots}
                        {
                            Follow \textsc{PartitionType\_IV} to find a target point \\
                            Change $r.color$ to \texttt{FINISH2} \\
                            Move to the target point
                        }
                    \Else   
                        {
                            No change in color and position
                        }
                }
            
        }

	\caption{\textsc{Square\_Partition}}
 \label{squarepartition}
\end{algorithm2e}

\begin{algorithm2e}
\setcounter{AlgoLine}{29}
\ElseIf{$r.color =$ \texttt{MONITOR}}
        {   
            \If{$r$ lies on $S_r$}
                {
                    Change $r.color$ to \texttt{OFF}
                }
            \Else
                {
                    \If{$r$ finds any \texttt{OFF}-colored robot in $Int(\Re)$}
                        {
                            $r$ does not change its color and position
                        }
                    \Else
                        {
                            Follow \textsc{MonitorToFinal\_inSqaure}() to find the target point and move to it
                        }
                }
      
        }

    \Else(\tcp*[h]{$r.color =$ \texttt{FINISH} or \texttt{FINISH1} or \texttt{FINISH2}})
        {
            $r$ terminates
        }

\end{algorithm2e}

\begin{algorithm2e}[!ht]
    \If{$r.color =$ \texttt{OFF}}
        {
            \If{$r$ is an interior robot}
                {
                    Follow \textsc{MoveToBoundary\_inCircle}() to find a target point on the boundary\\

                    Move to the target point with the color \texttt{OFF}
                    
                } 
            \ElseIf{$r$ is a boundary robot with at least one \texttt{OFF}-colored robot in $Int(\Re)$}
                {
                    $r$ does not change its color or position
                }
            \Else(\tcp*[h]{$r$ is a boundary robot with no \texttt{OFF}-colored robot in $Int(\Re)$})
                {
                    $Cl \longleftarrow$ The cluster of $r$ \\
                    \If{$Cl$ is eligible}
                        {
                            \If{$r$ is the head of $Cl$}
                                {
                                    Change $r.color$ to \texttt{HEAD}
                                }
                            \ElseIf{$r$ is the tail of $Cl$}
                                {
                                    Change $r.color$ to \texttt{TAIL}
                                }
                            \Else
                                {
                                    Change $r.color$ to \texttt{MID}
                                }
                        }
                    \Else   
                        {
                            $r$ does not change its position or color
                        }
                    
                }
        }

	\caption{\textsc{Circle\_Partition}}
 \label{circlepartition}
\end{algorithm2e}

\begin{algorithm2e}
\setcounter{AlgoLine}{17}
    \ElseIf{$r.color = $  \texttt{HEAD}}
        {
            $r' \longleftarrow$ Head of the neighbouring cluster $Cl'$ and the neighbour of $r$\\

            \If{($r'.color = $ \texttt{HALF} and $Cl'$ is eligible) $\lor$ ($Cl'$ is non-eligible)}
                {
                    Change $r.color$ to  FULL\\
                    Move to a point $t_r$ on $\wideparen{rr'}$ such that $alen(rt_r) = alen(rr')-s$
                }
            \Else(\tcp*[h]{$r'.color =$ \texttt{OFF} or \texttt{HEAD} with $Cl'$ being eligible})   
                {
                    Change $r.color$ to \texttt{MOVE-H}\\
                    Move to a point $t_r$ on $\wideparen{rr'}$ such that $alen(rt_r) = \frac{1}{2}(alen(rr')-s)$
                }
        }
\ElseIf{$r.color = $ \texttt{MOVE-H}}
        {
            Change $r.color$ to \texttt{HALF}
        }
    \ElseIf{$r.color =$ \texttt{MID}}
        {
            $r_{Nbr_1}, r_{Nbr_2} \longleftarrow$ Neighbours of $r$ in $Cl$ with $r_{Nbr_1}$ nearer to the head of $Cl$ than $r_{Nbr_2}$\\

            \If{$alen(rr_{Nbr_1}) = alen(rr_{Nbr_2})  =s$}
                {
                    $r$ does not change its color or position\\
                    
                }
            \Else
                {
                    $r$ moves to a point $t_r$ on $\wideparen{rr_{Nbr_1}}$ such that $alen(rt_r) = alen(rr_{Nbr_1})-s$
                }
            
        }
    \ElseIf{$r.color =$ \texttt{TAIL}}
        {
            $r_{Nbr_1}\longleftarrow$ Neighbour of $r$ in $Cl$\\
            \If{$alen(rr_{Nbr_1})=s$}
                {
                    $r$ does not change its color or position
                }
            \Else
               {    
                    Move to a point $t_r$ on $\wideparen{rr_{Nbr_1}}$ such that $alen(rt_r) = alen(rr_{Nbr_1})-s$ with color \texttt{TAIL}
               } 
        }
\ElseIf{$r.color = $ \texttt{HALF} or \texttt{FULL}}
            {
                $r_{Nbr_1}, r_{Nbr_2} \longleftarrow$ The two neighbours of $r$\\
                
               \If{$alen(rr_{Nbr_1}) = alen(rr_{Nbr_2}) = s$}
                    {
                        Change $r.color$ to $\texttt{OFF}$
                    }
                \Else
                    {
                        $r$ does not change its color or position
                    }
            }
        \Else(\tcp*[h]{$r.color = $ \texttt{FINSH} and $r$ lies on $\frac{1}{2}S_r$ or $\frac{1}{4}S_r$})
            {
                $r$ terminates.
            }
\end{algorithm2e}

\begin{algorithm2e}[!ht]
     \If{$r$ is an interior robot}
                {
                    Choose one of the nearest longest side of $\Re$ as $S_r$
                }
            \ElseIf{$r$ is a corner robot}
                {
                    Choose the incident longest side of $\Re$ as $S_r$
                }
            \ElseIf{$r$ is a boundary robot on one of the shortest sides of $\Re$}
                {
                    Choose one of the nearest longest sides of $\Re$ as $S_r$
                }
            \Else
                {
                    Choose the side where it is currently situated as $S_r$
                }
            
	\caption{\textsc{Rectangle\_Target\_Side}()}
 \label{rectangletargetside}
\end{algorithm2e}

\begin{algorithm2e}[!ht]
        $\mathcal{V}_r \longleftarrow$ The set of all robots not lying on $L_r$\\
        \If{the target side of $r$ is $S_r$}
            {
                \If{$p_r$ is not visible to $r$}
                {
                       No change in color and position
                }
            \Else
                {
                    \If{($r$ is a boundary robot not lying on one of the sides $S_r$ and $S_r^{opp}$ of $\Re$ with no robot on $p_r$) $\lor$ ($r$ is an interior robot with another robot on $p_r$) $\lor$ ($r$ is a corner robot)}
                        {
                            
                            \If{$\mathcal{V}_r$ is empty}
                                {
                                    Choose the point $t_r$ as the target point on the side $S_r$ such that $d(p_r, t_r) = \frac{1}{2} \max \{ d(e^1_{S_r},p_r), d(e^2_{S_r}, p_r) \}$
                                }
                            \Else
                                {
                                    Choose the target point $t_r$ on $S_r$  such that $d(p_r,t_r) = \frac{1}{4} \min\limits_{r' \in \mathcal{V}_r} \{ d(r',L_r) \}$
                                }
                            
                        }
                    \ElseIf{$r$ is an interior robot with no robot on $p_r$}
                        {
                            $r$ chooses the point $p_r$ as the target point
                        }
                    \Else
                        {
                            $r$ does not move or change its current color
                        }

                }
            }

	\caption{\textsc{TargetPoint\_On\_Sides}()}
 \label{targetpointonsides_subroutine}
\end{algorithm2e}

\begin{algorithm2e}
\setcounter{AlgoLine}{14}
        \Else
            {
                $p_r^{opp} \longleftarrow$ The point of intersection $L_r$ and $S_r^{opp}$ \\

                    \If{$p_r^{opp}$ has no robots on it}
                        {
                            Choose the point $p_r^{opp}$ as the target point
                        }
                    \Else
                        {
                            
                            Choose a point $t_r$ on $S_r^{opp}$ as the target point such that $d(p_r^{opp}, t_r) = \frac{1}{4} \min\limits_{r' \in \mathcal{V}_r} {d(r', L_r)}$ \\
                            
                        }
            }
\end{algorithm2e}

    \begin{algorithm2e}[!ht]
            \If{$S_r^{opp}$ has some robots on it}
                {
                    $r_{Nbr} \longleftarrow$ The neighbour of $r$ on $\mathcal{CH}_r$ lying on $(S_r^{opp}, \frac{1}{8}S_r^{opp}]$\\

                    Choose the point of intersection of $\frac{1}{8}S_r$ and the line $\overleftrightarrow{rr_{Nbr}}$ as the target point $a_r$
                }
            \Else
                {
                    Choose the point of intersection of $\frac{1}{8}S_r$ and $L_r$ as the target point $a_r$ 
                }
            
	\caption{\textsc{Move\_To\_ApexPoint}()}
 \label{movetoapexpoint_subroutine}
\end{algorithm2e}

   \begin{algorithm2e}[!ht]

            \If{$r.color=$ \texttt{FINISH}}
                {
                    \If{$0 < c_r^1 = c_r^2$}
                        {
                            $r_{Nbr} \longleftarrow$ The neighbour of $r$ on $\mathcal{CH}_r$ whose nearest longest side is $S_r^{opp}$ \\
                            
                            $SS_r \longleftarrow$ The shortest side of $\Re$ which either lies on or intersects the half plane delimited by the line $rr_{Nbr}$ where no other robot in $\Re$ is present\\

                            Compute the target point $t_r$ on $\frac{1}{4}S_r$ such that $d(t_r, SS_r) = \frac{len(S_r)}{2c_r^1}$\\
                        }
                    \Else(\tcp*[h]{$0 = c_r^2 < c_r^1$})
                        {
                            $SS_r \longleftarrow$ The shortest side of $\Re$ which lies on the half plane delimited by $L_r$ where no other robot in $\Re$ resides
                            
                            Calculates the target point $t_r$ on $\frac{1}{2}S_r$ such that $d(t_r, SS_r) = \frac{len(S_r)}{2c_r^1}$\\
                        }
                }

	\caption{\textsc{Move\_To\_FinalPoint}()}
 \label{movetofinalpoint}
\end{algorithm2e}

\begin{algorithm2e}
\setcounter{AlgoLine}{8}

 \Else
                {
                    $SS_r \longleftarrow$  The side of $\Re$ for which the intersection point of $S_r$ and $SS_r$ is visible to $r$\\

                    $SS_r^{opp} \longleftarrow$ The side of $\Re$ opposite to $SS_r$ 
                    
                    \If{$r$ finds a \texttt{FINISH}-colored robot on $\frac{1}{4}S_r$ but no robot on $[S_r^{opp}, \frac{1}{4}S_r^{opp}]$}
                        {
                            $r$ does not change its color and position
                        }
                    \Else   
                        {

                            \If{there is a \texttt{FINISH}-colored robot on $\frac{1}{2}S_r$}
                                {
                                    Choose $L = \frac{1}{2}S_r$
                                }
                            \Else(\tcp*[h]{(a \texttt{FINISH}-colored robot on $\frac{1}{4}S_r$ and some robots on $[S_r^{opp}, \frac{1}{4}S_r^{opp}]$) or (a \texttt{FINISH}-colored robot on $\frac{1}{4}S_r^{opp}$, but not on $[S_r, \frac{1}{4}S_r]$)})
                                {
                                    Choose $L = \frac{1}{4}S_r$
                                }
                            $\mathcal{F}_r \longleftarrow$ Set of all \texttt{FINISH}-colored robots\\
                            
                            Calculate $d = \min\limits_{r' \in \mathcal{F}_r} \{ \min \{d(r', SS_r), d(r', SS_r^{opp})\} \}$ \\
                            
                            $r_1, r_2, \cdots, r_k$ is the sequence of robots of maximum length on $L$ starting from $r_1$ with $d(r_1, SS_r) = d$ till $r_k$ such that two consecutive robots in the sequence are exactly at $2d$ distance apart from each other\\
        
                            Choose the target point $t_r$ on $L$ such that $d(t_r, SS_r) = (2k+1)d$
                        }
                }
\end{algorithm2e}

            

  \begin{algorithm2e}[!ht]

    \If{$r$ does not find any robot on $\Delta S_r^{opp}$ nor $\Delta S_r^{L}$ nor $\Delta S_r^R$}
        {
            $r$ chooses the point of intersection of $L_r$ and $\frac{1}{8}S_r$ as the apex point $a_r$
        }
    \Else
        {
            $r_{Nbr_1} \longleftarrow $ The neighbor of $r$ on $\mathcal{CH}_r$ that does not lie on $S_r$\\
    
            $r_{Nbr_2} \longleftarrow$ Another neighbor of $r$ on $\mathcal{CH}_r$ that lies on $S_r$\\
        
            Compute the target point $t_r$ on $\overline{rr_{Nbr_1}}$ such that $d(r,t_r) = \frac{1}{2} \min \{d(r,D_r), d(r, D_r^{opp}), d(r,\frac{1}{4}S_r), d(r, r_{Nbr_2}) \}$
        }      
	\caption{\textsc{Apex\_Point\_inSqaure}()}
 \label{apexpointinsquare}
\end{algorithm2e}

\begin{algorithm2e}[!ht]

    \If{$r$ finds \texttt{FINISH}-colored robots and all of them lying on $\frac{1}{2}S_r$, or $\frac{1}{4}S_r$ or $\frac{3}{4}S_r$}
        {
            Follow \textsc{Move\_To\_FinalPoint}() to find a target point and
            set $r.color=$ \texttt{FINISH}
            
        }     
    \ElseIf{$r$ finds \texttt{FINISH1}-colored robots and all of them lying on $\frac{1}{3}S_r$ or $\frac{1}{3}S$ where $S \in \{S_r^L, S_r^R \}$}
        {
            Set $k =\frac{1}{3}$ and follow \textsc{PartitionType\_III}() to get the target point\\
            Change $r.color $ to \texttt{FINISH1}
        }
    \ElseIf{$r$ sees \texttt{FINISH2}-colored robots and all of them lying  on $\frac{1}{6}S_r$ or on $\frac{1}{6}S$ where $S \in \{S_r^{opp}, S_r^L  S_r^R\}$}
        {
            Set $k =\frac{1}{6}$ and follow \textsc{PartitionType\_IV}() to get the target point \\
            Change $r.color$ to \texttt{FINISH2}
        }
        
    \ElseIf{($r$ finds a \texttt{FINISH}-colored robot, but not on $\frac{1}{2}S_r$, $\frac{1}{4}S_r$ or $\frac{3}{4}S_r$) $\lor$ ($r$ finds a \texttt{FINISH1}-colored robot, but not on $\frac{1}{3}S_r$ or $\frac{1}{3}S$ where $S \in \{S_r^L, S_r^R \}$) $\lor$ ($r$ finds a \texttt{FINISH2}-colored robot, but not on $\frac{1}{6}S_r$ or on $\frac{1}{6}S$ where $S \in \{S_r^{opp}, S_r^L  S_r^R\}$)}
        {
            No change in color or position
        }
    
    \Else
        {
            Follow \textsc{FinalPosition\_with$Max_r$}()
        }
	\caption{\textsc{MonitorToFinal\_inSqaure}()}
 \label{monitor_to_final_position}
\end{algorithm2e}

\begin{algorithm2e}[!ht]
    $Max_r \longleftarrow$ The set of all sides with maximum number of robots lying on their corresponding side triangle\\

    \If{$Max_r = \{ S_r\}$ and $\Delta S$ has no robots for all $S\in \{S_r^{opp}, S_r^L, S_r^R$ \}}
        {
            Follow \textsc{PartitionType\_I\_II}() with $k= \frac{1}{2}$ to find the target point and change $r.color$ to \texttt{FINISH}
        }
    \ElseIf{$Max_r = \{S_r, S_r^{opp}\}$ and no robots on $\Delta S$ for all $S \in \{S_r^L, S_r^R \}$}
        {
            Follow \textsc{PartitionType\_I\_II}() with $k=\frac{1}{4}$ to find the target point and change $r.color$ to \texttt{FINISH}\\
            
        }
    \ElseIf{$Max_r = \{ S_r, S \}$ and both $\Delta S_r^{opp}$ and $\Delta S^{opp}$ contain no robot where $S \in \{S^L_r, S^R_r\}$}
        {
            Follow \textsc{PartitionType\_III}() with $k = \frac{1}{3}$ to find the target point and change $r.color$ to \texttt{FINISH1}
        }
    \ElseIf{$Max_r = \{ S_r, S_r^{opp}, S_r^L, S^R_r \}$}
        {
            Follow \textsc{PartitionType\_IV}() with $k = \frac{1}{6}$ to find the target point and change $r.color$ to \texttt{FINISH2}
        }

	\caption{\textsc{FinalPosition\_with$Max_r$}()}
 \label{final_with_max_r}
\end{algorithm2e}

\begin{algorithm2e}
\setcounter{AlgoLine}{9}
    \Else
        {
            
                    \If{($Max_r = \{ S_r, S_r^{opp}, S\}$ where $S \in \{S_r^L , S^R_r\}$) $\lor$ ($Max_r = \{ S \}$ or $\{S, S'\}$, where $S, S' \in \{S_r^L , S^R_r\}$ with $S \neq S'$) $\lor$ ($Max_r = \{ S, S_r^{opp}\}$, where $S \in \{S_r^L , S^R_r\})$)}
                        {
                            Set $S$ as the target side $S_{Tgt}$
                        }
                    
                    \ElseIf{($Max_r = \{ S_r \}$ or $\{ S_r, S_r^L, S_r^R \}$) $\lor$ ($Max_r = \{S, S_r\}$ where $S \in \{S_r^{opp}, S_r^L , S^R_r\}$)}
                        {
                            Set $S_r$ as the target side $S_{Tgt}$
                        }
                    \Else(\tcp*[h]{$Max_r =  \{ S_r^L, S^R_r, S_r^{opp}\}$ or \{$S_r^{opp} \})$})
                        {
                            Set $S_r^{opp}$ as the target side $S_{Tgt}$
                        }
            
                    \If{$S_{Tgt} = S_r$}
                        {
                            $r_{Nbr_1}\longleftarrow$ The neighbour of $r$ on $\mathcal{CH}_r$ not lying on $S_r$\\
                            Select the point of intersection of $\overleftrightarrow{rr_{Nbr_1}}$ and $S_r$ as the target point $t_r$ 
                        }
                    \Else
                        {
                            $L_r' \longleftarrow$ The line parallel to $S_{Tgt}$ passing through $r$\\
                            \If{ there is a \texttt{MONITOR}-colored robot lying on the half plane delimited by $L_r'$ that contains $S_{Tgt}$}
                                {
                                    $r$ does not change its color or position
                                }
                            \Else
                                {
                                    $L_r^\perp \longleftarrow$ The line passing through $r$ and perpendicular to $S_{Tgt}$\\
                                    $p_{r_{Tgt}} \longleftarrow$ The point of intersection of $L_r^\perp$ and $S_{Tgt}$\\
                                    \If{$p_{r_{Tgt}}$ has no robots on it}
                                        {
                                            Choose the point $p_{r_{Tgt}}$ as the target point
                                        }
                                    \Else
                                        {
                                            Choose a point $t_r$ on $S_{Tgt}$ as the target point such that $d(p_{r_{Tgt}}, t_r) = \frac{1}{4}\min\limits_{r' \in \mathcal{V}_r} d(r', L_r^\perp)$
                                        }
                                    Change $r.color$ to \texttt{OFF}
                                }
                        }

        }
\end{algorithm2e}

\begin{algorithm2e}[!ht]
$c_r^1 \longleftarrow$ The number of robots on $\Delta S_r$\\

    \If{$k = \frac{1}{2}$}
        {
            $SS_r\longleftarrow$ The side which lies on the half plane delimited by $L_r$ where no other robot resides\\
            Calculate the target point $t_r$ on $\frac{1}{2}S_r$ such that $d(t_r, SS_r)= \frac{len(S_r)}{2c_r^1}$ 
        }
    \Else
        {
            $r_{Nbr}\longleftarrow$ The neighbour of $r$ on $\mathcal{CH}_r$ whose nearest side is $S_r^{opp}$\\
            $SS_r \longleftarrow$ The side which either lies on or intersects the half plane delimited by the line $rr_{Nbr}$ where no other robot is present\\
            Compute the target point $t_r$ on $\frac{1}{4}S_r$ such that $d(t_r, SS_r) = \frac{len(S_r)}{2c_r^1}$
        }

	\caption{\textsc{PartitionType\_I\_II}()}
 \label{typeIandII}
\end{algorithm2e}

\begin{algorithm2e}[!ht]
    \If{$k = \frac{1}{3}$ and $r.color =$ \texttt{MONITOR}}
        {
            $c \longleftarrow$ The number of all \texttt{FINISH1}-colored robots on $\frac{1}{3}S_r$ \\
            $c' \longleftarrow$ The number of robots lying on $\Delta S_r$ without the color \texttt{FINISH1} \\
            \If{$S = S_r^L$}
                {
                     $A,B \longleftarrow$ The two points on $S_r$ which are $\frac{len(S_r)}{c+ c'}$ distance away from $e_P$ and $e_Q$, respectively \\

                    \If{$S^{opp}$ lies on the half plane delimited by $L_r$ where other robots on $\Delta S_r$ lie}
                        {
                            Choose the triangle $\mathcal{T} = \Delta Ae_Pe_R$
                        }
                    \Else
                        {
                            Choose $\mathcal{T} = \Delta Be_Qe_R$
                        }
                }
            \Else   
                {
                    $A,B \longleftarrow$ The two points on $S_r$ which are $\frac{len(S_r)}{c+ c'}$ distance away from $e_Q$ and $e_P$, respectively \\

                    \If{$S^{opp}$ lies on the half plane delimited by $L_r$ where other robots on $\Delta S_r$ lie}
                        {
                            Choose the triangle $\mathcal{T} = \Delta Ae_Qe_X$
                        }
                    \Else
                        {
                            Choose $\mathcal{T} = \Delta Be_Pe_X$
                        }
                }
           
            Set the centroid of the triangle $\mathcal{T}$ as the target point
        }

	\caption{\textsc{PartitionType\_III}()}
 \label{typeIII}
\end{algorithm2e}

\begin{algorithm2e}\small
\setcounter{AlgoLine}{16}
    \Else(\tcp*[h]{$k = \frac{1}{3}$ and $r.color = $ \texttt{OFF}})
        {
            \If{($\frac{1}{3}S_r^L$ has a \texttt{FINISH1}-colored robot) $\lor$ (a robot lying on $\Delta S_r^L$)}
                {
                    Choose $S = S_r^L$
                }
            \Else
                {
                    Choose $S = S_r^R$
                }
            \If{there is a \texttt{FINISH1}-colored robots not lying on $\frac{1}{3}S_r$ or on $\frac{1}{3}S$}
                {
                    No change in color or position
                }
            \Else
                {
                    Choose the point of intersection of $S_r^{opp}$ and $S^{opp}$ as the vertex $v_r$\\
                    $D_r \longleftarrow$ Diagonal of $\Re$ passing through $v_r$ \\
                    
                    \If{there is a \texttt{FINISH1}-colored robot on $\frac{1}{3}S_r$}
                        {
                            $r'\longleftarrow$ Terminal \texttt{FINISH1}-colored robot on $\frac{1}{3}S_r$ \\
                            $bl_1,  bl_2\longleftarrow$ The length of the base  of the triangle  whose one side is $D_r$ (resp. $S^{opp}$) and the centroid is at $r'$ \\
                            
                        }
                    \Else
                        {
                            $r'\longleftarrow$ Terminal \texttt{FINISH1}-colored robot on $\frac{1}{3}S$ \\
                            $bl_1, bl_2 \longleftarrow$ The length of the base  of the triangle  whose one side is $D_r$ (resp. $S_r^{opp}$) and the centroid is at $r'$ \\
                            
                        }
                    $bl = \min\{ bl_1, bl_2 \}$\\
                    $u_1 \longleftarrow$ The point of intersection of $\frac{1}{3}S_r$ and the diagonal $D_r$ \\
                    $u_2 \longleftarrow$ The point of intersection $\frac{1}{3}S_r$ and the diagonal $S^{opp}$ \\
                    $e_P \longleftarrow$ The endpoint of $D_r$ other than $v_r$\\
                    $e_Q \longleftarrow$ The endpoint of $S_r$ other than $e_P$ \\

                    \If{$S^{opp}$ lies in the half plane delimited by $\overleftrightarrow{rv_r}$, where the other robots of $S_r$ reside}
                        {
                            $c \longleftarrow$ The number of \texttt{FINISH1}-colored robots on $\frac{1}{3}S_r$ starting from the robot $\frac{bl}{3}$ distance apart from $u_1$ towards $u_2$ such that two consecutive robots are $\frac{2bl}{3}$ distance away from each other \\
                            $e_Y, e_Z \longleftarrow$ The points of $S_r$ satisfying $d(e_Y, e_P) = c\cdot bl$ and $d(e_Z, e_P) = (c+1)\cdot bl$\\

                        }
                    \Else
                        {
                            $c\longleftarrow$ The number of \texttt{FINISH1}-colored robots on $\frac{1}{3}S_r$ starting from the robot $\frac{bl}{3}$ distance apart from $u_2$ towards $u_1$ such that two consecutive robots are $\frac{2bl}{3}$ distance away from each other \\
                            $e_Y, e_Z \longleftarrow$ The points of $S_r$ satisfying $d(e_Y, e_Q) = c\cdot bl$ and $d(e_Z, e_Q) = (c+1)\cdot bl$\\
                        }
                    Calculate the target point $t_r$ such that it is the centroid of $\Delta e_Ye_Zv_r$ 
                }
        }
\end{algorithm2e}

\begin{algorithm2e}[!ht]\small
    
    \If{$k = \frac{1}{6}$ and $r.color =$ \texttt{MONITOR}}
        {
            $A, B \longleftarrow$ The two points on $S_r$ which are $\frac{len(S_r)}{|S_r|}$ distance away from $e_P$ and $e_Q$, respectively \\
            \If{$e_P$  and all the other robots on $\Delta S_r$ lie on the different half plane delimited by $L_r$}
                {
                    Choose the triangle $\mathcal{T} = \Delta Ae_PC$
                }
            \Else
                {
                    Choose $\mathcal{T} = \Delta Be_QC$
                }
            $r$ chooses the centroid of $\mathcal{T}$ as the target
        }
    \Else(\tcp*[h]{$k = \frac{1}{6}$ and $r.color =$ \texttt{OFF}})
        {
            \If{there is a \texttt{FINISH2}-colored robot not lying on $\frac{1}{6}S$ for $S \in \{S_r, S_r^{opp}, S_r^L, S_r^R \}$}
                {
                    No change in color and position
                }
            \Else
                {
                    $r' \longleftarrow$ A \texttt{FINISH2}-colored terminal robot on $\frac{1}{6}S_{r'}$ \\
                    $e^1_{S_{r'}}, e^2_{S_{r'}} \longleftarrow$ Two endpoints of $S_{r'}$ \\
                    $D_r, D_r^{opp} \longleftarrow$ Two diagonals of $\Re$ passing thorugh $e_P$ and $e_Q$ respectively\\
                    $bl_1 \longleftarrow$ Length of the base of the triangle whose centroid is $r'$ and two vertices are $C$ and $e^1_{S_{r'}}$ \\
                    $bl_2 \longleftarrow$ Length of the base of the triangle whose centroid is $r'$ and two vertices are $C$ and $e^2_{S_{r'}}$ \\
                    $bl = \min\{bl_1, bl_2\}$ \\
                    $u_1 \longleftarrow$ The point of intersection of $\frac{1}{6}S_r$ and $D_r$ \\
                    $u_2 \longleftarrow$ The point of intersection of $\frac{1}{6}S_r$ and $D_r^{opp}$ \\

                    \If{$\overline{Ce_Q}$ lies on the half plane delimited by $\overleftrightarrow{rC}$ where other robots of $S_r$ reside}
                        {
                            $c \longleftarrow$ The number of \texttt{FINISH2}-colored robots on $\frac{1}{6}S_r$ starting from the robot ${bl}/3$ distance away from $u_1$ towards $u_2$ such that two consecutive robots are ${2bl}/3$ distance apart from each other\\

                            $e_Y, e_Z \longleftarrow$ The points on $S_r$ satisfying $d(e_P, e_Y) = c\cdot bl$ and $d(e_P, e_Z) = (c+1)\cdot bl$\\

                            Choose the centroid of $\Delta Ce_Ye_Z$ as the target point
                        }
                    \Else
                        {
                            $c \longleftarrow$ The number of \texttt{FINISH2}-colored robots on $\frac{1}{6}S_r$ starting from the robot ${bl}/3$ distance away from $u_2$ towards $u_1$ such that two consecutive robots are ${2bl}/3$ distance apart from each other\\

                            $e_Y, e_Z \longleftarrow$ The points on $S_r$ satisfying $d(e_Q, e_Z) = c\cdot bl$ and $d(e_Q, e_Y) = (c+1)\cdot bl$\\

                            Choose the centroid of $\Delta Ce_Ye_Z$ as the target point
                        }

                }
        }
    
	\caption{\textsc{PartitionType\_IV}()}
 \label{typeIV}
\end{algorithm2e}

\begin{algorithm2e}[!ht]
   \If{$r$ lies on $O$ and there is at least interior robot other than $r$}
        {
            $r$ does not change its color or position
        }
    \ElseIf{$r$ lies in $Int(\Re)
    $, but not on $O$}
        {
            $p_r \longleftarrow$ The point of intersection of the line $\overleftrightarrow{rO}$ and the boundary of $\Re$ \\
            $p_r^{opp}\longleftarrow$ The point diametrically opposite to $p_r$ on the boundary of $\Re$\\

            \If{$p_r$ does not have any robot on it}
                {
                    $r$ chooses $p_r$ as its target point
                }
            \Else
                {
                    $\mathcal{V}_r \longleftarrow$ The set of all visible robots to $r$ not lying on the line $\overleftrightarrow{rO}$\\

                    \If{$\mathcal{V}_r$ is non-empty}
                        {
                            Calculate $d_r = \frac{1}{4}\min\{alen(p_rp_{r'})| r' \in \mathcal{V_r} \text{ and } \wideparen{p_rp_{r'}} \text{ is defined}\}$   
                        }
                    \Else
                        {
                            Calculate $d_r = \frac{1}{4}alen(p_rp_r^{opp})$
                        }
                    $r$ finds a target point on the boundary $t_r$ such that $alen(t_rp_r) = \frac{d_r}{rad}d(r,p_r)$
                }
        }

	\caption{\textsc{MoveToBoundary\_inCircle}()}
 \label{movetoboundary_incircle}
\end{algorithm2e}

\end{document}